\documentclass[10pt,journal]{IEEEtran}
\ifCLASSOPTIONcompsoc
\else
\fi
\ifCLASSINFOpdf
\else
\fi
\usepackage{amsmath}
\usepackage{amsthm}
\usepackage{amssymb}
\usepackage[ruled,lined]{algorithm2e}
\usepackage{graphicx}
\usepackage{subfigure}
\usepackage{tikz}
\usepackage{bm}
\usepackage{booktabs}
\usepackage{tabularx}
\usepackage{caption}
\DeclareMathOperator*{\argmax}{arg\,max}

\begin{document}
%
\title{A	 Joint Optimization of Operational Cost and Performance Interference in Cloud Data Centers}
%
%
%
%

\author{\IEEEauthorblockN{Xibo Jin, Fa Zhang, Lin Wang, Songlin Hu, Biyu Zhou and Zhiyong Liu}\\
\IEEEauthorblockA{Institute of Computing Technology, Chinese Academy of Sciences\\
University of Chinese Academy of Sciences, Beijing, China\\
Email: \{jinxibo, zhangfa, wangling, husonglin, zhoubiyu, zyliu\}@ict.ac.cn}
}

\maketitle

\begin{abstract}
Virtual machine (VM) scheduling is an important technique to efficiently operate the computing resources in a data center. Previous work has mainly focused on consolidating VMs to improve resource utilization and thus to optimize energy consumption. However, the interference between collocated VMs is usually ignored, which can result in very worse performance degradation to the applications running in those VMs due to the contention of the shared resources. Based on this observation, we aim at designing efficient VM assignment and scheduling strategies where we consider optimizing both the operational cost of the data center and the performance degradation of running applications and then, we propose a general model which captures the inherent tradeoff between the two contradictory objectives. We present offline and online solutions for this problem by exploiting the spatial and temporal information of VMs where VM scheduling is done by jointly consider the combinations and the life-cycle overlapping of the VMs. Evaluation results show that the proposed methods can generate efficient schedules for VMs, achieving low operational cost while significantly reducing the performance degradation of applications in cloud data centers.
\end{abstract}



%
\IEEEpeerreviewmaketitle

\section{Introduction}


Cloud computing has become a promising choice for modern computing platforms and will most likely continue to be the dominant service model in the future. The foundation of cloud computing is founded by taking advantage of virtualization technologies such as VMware \cite{VMware} and Xen \cite{Xen} to encapsulate applications into virtual machines (VMs) and allow independent applications to execute on the same physical server simultaneously. Furthermore, cloud computing affords users to obtain, configure, and deploy cloud services themselves using cloud service catalogues, without requiring the assistance of IT (Infrastructure Technology) \cite{Perera12}. The feasibility of VM consolidation and on-demand resource allocation offers an opportunity for cloud operators to multiplex resources among users and thus improve the operational cost, e.g., reducing the energy consumption.

However, although it brings better utilization to the cloud system, such kind of resource multiplexing is not always beneficial. When VMs are consolidated together, the performance interference between the VMs brought by the contention of shared resources such as last-level-cache, memory bus, network and disk bandwidth can not be ignored \cite{Govindan11, Chiang11, Mars11, Roytman13, Kim13, Verboven13}. As compared to running in a dedicated server, a VM has to compete on the shared resources with other VMs that collocated with it and thus the performance will be degraded even with the same resource reservation. While previous work is focused on analysis on application-level interference in single servers, we aim to study the assignment and scheduling of VMs to physical servers, mitigating the performance interference while optimizing the operational cost. We tackle this combinatorial problem of joint optimization by leveraging the specific structures of VM collocation.

\subsection{Performance Interference inside a Cloud Data Center}

It is necessary to provide efficient management of performance interference in order to guarantee the quality of service for tenants in a cloud data center. In general, the performance interference between VMs can be affected by the following two factors.

\textbf{VM combination.} 
Recent researches have analysed the resource contention for possible VM combinations and suggested to collocate those VMs that have less competition between shared resources \cite{Govindan11, Chiang11, Mars11, Roytman13, Kim13, Verboven13}. In order to quantify the overall performance interference between collocated VMs, we evaluated the performance degradation of VMs using SPECcpu 2006 benchmark \cite{SPEC06} where we assume each application executes in a virtual machine and runs on a physical core. We define the \textit{Performance Degradation Ratio} (PDR) of a VM as the increment of running time divided by the time used for the VM to be executed in a dedicated server. The statistical results are demonstrated in Table~\ref{tb:spatial}. As can be seen from the table, the PDR of $\mathtt{429.mcf}$ when being collocated with $\mathtt{470.lbm}$ is $62.08\%$ while it is $11.90\%$ when being collocated with $\mathtt{403.gcc}$. This reveals that different VM combinations lead to variable level of performance interference. As a result, VM placement can be done in an intelligent way such that the performance interference between VMs is minimized. Another observation is that the PDRs of VMs become larger with the increase of number of collocated VMs. For example, in Table~\ref{tb:spatial}, the PDRs of $\mathtt{470.lbm}$ and $\mathtt{403.gcc}$ are $10.06\%$ and $44.53\%$ respectively, while these values increase up to $16.29\%$ and $67.55\%$ when a third VM for $\mathtt{429.mcf}$ is launched simultaneously on the same physical server.

\renewcommand{\arraystretch}{1.3}
\begin{table*}[htbp]
\centering
\caption{\label{tb:spatial}Running times (and stretches in percentile) of applications (VMs) collocated in the same physical server. (For example, the 1st row means each application runs on a dedicated server and, the 3rd row means $\mathtt{403.gcc}$ and $\mathtt{429.mcf}$ collocate in a server, et al.)}
\begin{tabular}{c|c|c|c|c|c}\hline\hline
&$\mathtt{401.bzip2}$&$\mathtt{403.gcc}$&$\mathtt{429.mcf}$&$\mathtt{453.povray}$&$\mathtt{470.lbm}$\\\hline
$\mathtt{401.bzip2}$/$\mathtt{403.gcc}$/$\mathtt{429.mcf}$/$\mathtt{453.povray}$/$\mathtt{470.lbm}$&498&265&269&186&318\\\hline
$\mathtt{401.bzip2}$ + $\mathtt{470.lbm}$&642 (28.92)&--&--&--&358 (12.58)\\
\textbf{\textit{403.gcc + 429.mcf}}&--&{\bf 299 (12.83)}&{\bf 301 (11.90)}&--&--\\
\textbf{\textit{429.mcf + 470.lbm}}&--&--&{\bf 436 (62.08)}&--&{\bf 366 (15.09)}\\
$\mathtt{403.gcc}$ + $\mathtt{453.povray}$&--&270 (1.89)&--&193 (3.76)&--\\
$\mathtt{453.povray}$ + $\mathtt{470.lbm}$&--&--&--&201 (8.06)&326 (2.52)\\
\textbf{\textit{403.gcc + 470.lbm}}&--&{\bf 383 (44.53)}&--&--&{\bf 350 (10.06)}\\\hline
\textbf{\textit{403.gcc + 429.mcf + 470.lbm}}&--&{\bf 444 (67.55)}&{\bf 487 (81.04)}&--&{\bf 407 (16.29)}\\
$\mathtt{401.bzip2}$ + $\mathtt{429.mcf}$ + $\mathtt{470.lbm}$&725 (45.58)&--&482 (79.18)&--&404 (27.04)\\
$\mathtt{401.bzip2}$ + $\mathtt{403.gcc}$ + $\mathtt{429.mcf}$ + $\mathtt{470.lbm}$&778 (56.22)&495 (86.79)&538 (100.00)&--&466 (46.54)\\\hline\hline
\end{tabular}
\end{table*}

\textbf{Life-cycle overlapping.} 
It is a challenging problem to take into account the life cycles of VMs. On the one hand, overlapping the execution of VMs can improve the resource utilization of the system and thus reduce the marginal cost.\footnote{This refers to the static cost irrespective of the load of server incurred by always-on components such as idle-energy.} On the other hand, due to performance interference, reducing the overlap of the executions of VMs can mitigate performance degradation thus shortening the completion times of VMs. This can be verified by the results shown in Table~\ref{tb:temporal}. For example, when collocated with $\mathtt{470.lbm}$, $\mathtt{429.mcf}$, and $\mathtt{401.bzip2}$, $\mathtt{403.gcc}$ receives a considerable reduction on PDR from $86.79\%$ to $10.57\%$ with the lessening of the execution overlaps. As a consequence of performance interference, the neglect of life-cycle overlapping can result in more serious problems such as resource-reservation violation brought by the stretch on the execution duration of VMs. Moreover, the performance of some VMs will become unacceptably worse when the execution is always overlapped with other mutual-interference VMs and therefore, their performance is degraded all the time by collocated VMs.

\begin{table}[!t]
\centering
\caption{\label{tb:temporal}Running times (and stretches in percentile) of applications (VMs) collocated in the same physical server with different overlap times. (For example, the 1st grid means $\mathtt{429.mcf}$ and $\mathtt{403.gcc}$ collocate in a physical server with different overlap times. I.e., $+60$ means that $\mathtt{403.gcc}$ starts after $\mathtt{429.mcf}$ has run 60 unit times.)}
\begin{tabular}{p{1.0cm}p{1.32cm}p{1.32cm}p{1.32cm}p{1.32cm}}\hline\hline 
{\bf Apps}	&	0&	+60&		+120	&	+180\\\hline
$\mathtt{429.mcf}$&	301 (11.90)&	{\bf 	294 (9.29)}&	283 (5.20)&	277 (2.97)\\
$\mathtt{403.gcc}$&	299 (12.83)&	{\bf 	292 (10.19)}&286 (7.92)&	278 (4.91)\\\hline
$\mathtt{470.lbm}$&	404 (27.04)&		377 (18.55)&	354 (11.32)&	336 (5.66)\\
$\mathtt{429.mcf}$&	482 (79.18)&		440 (63.57)&	401 (49.07)&	360 (33.83)\\
$\mathtt{401.bzip2}$& 725 (45.58)&		643 (29.12)&	578 (16.06)&	523 (5.02)\\\hline
$\mathtt{470.lbm}$&	466 (46.54)&		407 (27.99)&	357 (12.26)&	337 (5.97)\\
$\mathtt{429.mcf}$&	538 (100.0)&		476 (76.95)&	418 (55.39)&	365 (35.69)\\
\textbf{\textit{403.gcc}}&	{\bf 495 (86.79)}&		{\bf 424 (60.00)}&	{\bf 353 (33.21)}&	{\bf 293 (10.57)}\\
$\mathtt{401.bzip2}$&	778 (56.22)&		658 (32.13)&	550 (10.44)&	510 (2.41)\\\hline\hline
\end{tabular}
\end{table}

\subsection{Tradeoff between Operational Cost and Performance Interference}

In general, operational cost refers to the daily expenditure caused by the operation of a cloud computing system, including electricity cost and system maintenance expenses. Among them, the electricity cost takes a dominant proportion \cite{Hamilton09, Amokrane2013}. As a consequence, achieving energy efficiency on servers can result in significant reduction on the operational cost of a data center. For this reason, we will use the term \emph{energy consumption} to refer to operational cost. Throughout the paper, we use both terms interchangeably. There has been a large body of work focused on improving the energy efficiency of single servers, such as Dynamic Voltage Frequency Scaling (DVFS) and powering down \cite{Albers10}. Based on the two fundamental mechanisms, researches have investigated to reduce the energy consumption of a cloud system using virtualization techniques such as VM consolidation to improve hardware utilization. However, while these methods can help reach the goal of energy conservation elegantly, very little attention has been paid on the accompanying side-effect, i.e., performance interference. Moreover, to the best of our knowledge, a quantitative analysis on the tradeoff between energy consumption and performance interference is almost completely missing in the literature, which is highly desired by cloud operators.

We study the VM assignment and scheduling problem for arbitrating between energy consumption and performance interference, i.e., reducing energy consumption while maintaining low performance degradation for VMs. On the one hand, ideally, the energy consumption is minimized when a minimum number of servers is used. This can be done by consolidating VMs and then turning idle servers into some power-saving mode (sleeping or power-off). The set of active servers is managed dynamically according to the workload. Consequently, the energy consumed by underutilized servers can be saved, as well as the corresponding cost incurred by power delivery and cooling infrastructure. 

On the other hand, VM consolidation can result in undesirable performance interference between VMs because of the contention in shared resources. This performance interference can stretch the execution durations of VMs to a large extent, which may bring unacceptable performance loss to user applications (and further result in Service-Level-Agreement violation). A simple example is illustrated in Fig.~\ref{fig:example_allocation}. It can be observed that the assignment shown in the right-side figure is better than the one shown on the left in terms of two aspects: $\mathit{i)}$ the performance of most VMs such as $vm_4$ is less degraded and $\mathit{ii)}$ the real-time accommodation of $vm_6$ becomes possible. (As shown in the part of the ellipses.) This also reveals that the two factors, VM combination and life-cycle overlapping, are coupled and mutually affected. Therefore, in order to arbitrate between energy consumption and performance interference, it is necessary to provide a careful design of VM consolidation where VMs are allocated with appropriate combinations and collocated VMs are scheduled with the the most favourable life-cycle overlapping.

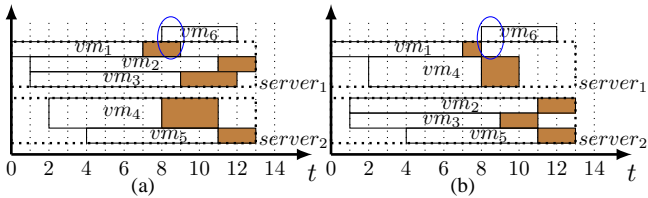
\begin{figure}
\begin{tikzpicture}
    \draw[->, >=latex,line width=1pt] (0.0,0)--(3.75+0.25,0) node[pos=1,below] {$t$};
    \draw[->, >=latex,line width=1pt] (0,0.0)--(0,2);
    \draw[dotted] (0.0,1.75)--(0.0,0.0) node[below, scale=0.8] {$0$};
    \draw[dotted] (0.25,1.75)--(0.25,0.0);
    \draw[dotted] (0.5,1.75)--(0.5,0.0) node[below, scale=0.8] {$2$};
    \draw[dotted] (0.75,1.75)--(0.75,0.0);
    \draw[dotted] (1,1.75)--(1,0.0) node[below, scale=0.8] {$4$};
    \draw[dotted] (1.25,1.75)--(1.25,0.0);
    \draw[dotted] (1.5,1.75)--(1.5,0.0) node[below, scale=0.8] {$6$};
    \draw[dotted] (1.75,1.75)--(1.75,0.0) node[below=6pt, scale=0.8] {(a)};
    \draw[dotted] (2,1.75)--(2,0.0) node[below, scale=0.8] {$8$};
    \draw[dotted] (2.25,1.75)--(2.25,0.0);
    \draw[dotted] (2.5,1.75)--(2.5,0.0) node[below, scale=0.8] {${10}$};
    \draw[dotted] (2.75,1.75)--(2.75,0.0);
    \draw[dotted] (3,1.75)--(3,0.0) node[below, scale=0.8] {${12}$};
    \draw[dotted] (3.25,1.75)--(3.25,0.0);
    \draw[dotted] (3.5,1.75)--(3.5,0) node[below, scale=0.8] {${14}$}; 
    
    \draw(0.0,1.5) rectangle (1.75,1.3) node[pos=0.5,text width=0.25em, scale=0.8] {$vm_{1}$};
    \draw[fill=brown](1.75,1.5) rectangle (2.25,1.3);
    
    \draw(2.0,1.7) rectangle (3,1.5) node[pos=0.5, scale=0.8] {$vm_{6}$};
    
    \draw[color=blue] (2.12,1.55) ellipse(5pt and 8pt);
    
    \draw(0.25,1.3) rectangle (2.75,1.1) node[pos=0.5,text width=0.25em, scale=0.8]{$vm_{2}$} ;
    \draw[fill=brown](2.75,1.3) rectangle (3.25,1.1);
    
    \draw(0.25,1.1) rectangle (2.25,0.9) node[pos=0.5,text width=0.25em, scale=0.8]{$vm_{3}$};
    \draw[fill=brown](2.25,1.1) rectangle (3,0.9);
        
    \draw(0.5,0.75) rectangle (2,0.35) node[pos=0.5,text width=0.25em, scale=0.8]{$vm_{4}$};
    \draw[fill=brown] (2,0.75) rectangle (2.75,0.35);
    
    \draw(1,0.35) rectangle (2.75,0.15) node[pos=0.5,text width=0.25em, scale=0.8]{$vm_{5}$};
    \draw[fill=brown] (2.75,0.35) rectangle (3.25,0.15);   
    
    \draw[dotted,line width=0.9pt](0.0,1.5) rectangle (3.25,0.9) node[right=0pt,pos=0.98,scale=0.8] {$server_1$};    
    \draw[dotted,line width=0.9pt](0.0,0.75) rectangle (3.25,0.15) node[right=0pt,pos=0.98,scale=0.8] {$server_2$}; 
    
    \draw[->, >=latex,line width=1pt] (4.0+0.25,0)--(7.75+0.25+0.25,0) node[pos=1,below] {$t$};
    \draw[->, >=latex,line width=1pt] (4+0.25,0.0)--(4+0.25,2);
    \draw[dotted] (4.0+0.25,1.75)--(4.0+0.25,0.0) node[below, scale=0.8] {$0$};
    \draw[dotted] (4.25+0.25,1.75)--(4.25+0.25,0.0);
    \draw[dotted] (4.5+0.25,1.75)--(4.5+0.25,0.0) node[below, scale=0.8] {$2$};
    \draw[dotted] (4.75+0.25,1.75)--(4.75+0.25,0.0);
    \draw[dotted] (5+0.25,1.75)--(5+0.25,0.0) node[below, scale=0.8] {$4$};
    \draw[dotted] (5.25+0.25,1.75)--(5.25+0.25,0.0);
    \draw[dotted] (5.5+0.25,1.75)--(5.5+0.25,0.0) node[below, scale=0.8] {$6$};
    \draw[dotted] (5.75+0.25,1.75)--(5.75+0.25,0.0) node[below=6pt, scale=0.8] {(b)};
    \draw[dotted] (6+0.25,1.75)--(6+0.25,0.0) node[below, scale=0.8] {$8$};
    \draw[dotted] (6.25+0.25,1.75)--(6.25+0.25,0.0);
    \draw[dotted] (6.5+0.25,1.75)--(6.5+0.25,0.0) node[below, scale=0.8] {${10}$};
    \draw[dotted] (6.75+0.25,1.75)--(6.75+0.25,0.0);
    \draw[dotted] (7+0.25,1.75)--(7+0.25,0.0) node[below, scale=0.8] {${12}$};
    \draw[dotted] (7.25+0.25,1.75)--(7.25+0.25,0.0);
    \draw[dotted] (7.5+0.25,1.75)--(7.5+0.25,0.0) node[below, scale=0.8] {${14}$};  
    
    \draw(4.0+0.25,1.5) rectangle (5.75+0.25,1.3) node[pos=0.5,text width=0.25em, scale=0.8] {$vm_{1}$};
    \draw[fill=brown](5.75+0.25,1.5) rectangle (6.0+0.25,1.3);
    
    \draw(4.5+0.25,1.3) rectangle (6+0.25,0.9) node[pos=0.5,text width=0.25em, scale=0.8]{$vm_{4}$};
    \draw[fill=brown] (6+0.25,1.3) rectangle (6.5+0.25,0.9);
    
    \draw(6.0+0.25,1.7) rectangle (7+0.25,1.5) node[pos=0.5, scale=0.8] {$vm_{6}$};
    
    \draw[color=blue] (6.0+0.25+0.12,1.55) ellipse(5pt and 8pt);
    
    \draw(4.25+0.25,0.75) rectangle (6.75+0.25,0.55) node[pos=0.5,text width=0.25em, scale=0.8]{$vm_{2}$} ;
    \draw[fill=brown](6.75+0.25,0.75) rectangle (7.25+0.25,0.55);
    \draw(4.25+0.25,0.55) rectangle (6.25+0.25,0.35) node[pos=0.5,text width=0.25em, scale=0.8]{$vm_{3}$};
    \draw[fill=brown](6.25+0.25,0.55) rectangle (6.75+0.25,0.35);
        
    \draw(5+0.25,0.35) rectangle (6.75+0.25,0.15) node[pos=0.5,text width=0.25em, scale=0.8]{$vm_{5}$};
    \draw[fill=brown] (6.75+0.25,0.35) rectangle (7.25+0.25,0.15);
    
    \draw[dotted,line width=0.9pt](4.0+0.25,1.5) rectangle (4.0+0.25+3.25,0.9) node[right=0pt,pos=0.98,scale=0.8] {$server_1$};    
    \draw[dotted,line width=0.9pt](4.0+0.25,0.75) rectangle (4.0+0.25+3.25,0.15) node[right=0pt,pos=0.98,scale=0.8] {$server_2$}; 
 
\end{tikzpicture}
\caption{\label{fig:example_allocation}Two ways of VM allocation. Assume the two servers have the same resource capacity of \{$\bm{1}$\} unit and each VM from $\{vm_1,vm_2,vm_3,vm_5,vm_6\}$ demands \{$\bm{\frac{1}{3}}$\} unit while $vm_4$ requires \{$\bm{\frac{2}{3}}$\} unit. The rectangles in color represent the stretch of execution time due to resource contention, which represent the same mean in the following figures. (a): An inappropriate scheduling, (b): A better scheduling.}
\end{figure}

\subsection{Overview of the Paper}

In this work, we seek to find out efficient solutions for reducing the energy consumption while minimizing the performance interference among VMs. Our main contributions are summarized in the following three aspects:
\begin{enumerate}
\item We characterize the energy consumption and the performance interference in a unified model and  formally formulate the challenge of VM assignment and scheduling into an optimization problem. We also prove the NP-Completeness of the problem;
\item We propose efficient algorithms for offline VM assignment and scheduling, assuming all information is known a priori;
\item We extend the offline algorithms to the case with dynamic VM arrival. Using information such as resource reservation, these algorithms can further be improved. We also provide a distributed implementation of the algorithms for large-scale data centers; and
\item We evaluate the efficiency of the proposed algorithms through comprehensive simulations, showing that the proposed solution can achieve desirable arbitration between energy consumption and performance interference.
\end{enumerate}

The rest of this paper is organized as follows. In Section~\ref{sec:related}, we summarize the related works relevant to ours. Section~\ref{sec:model} shows the modelling of the problem. Section~\ref{sec:algorithm} provides our algorithms for VM assignment and scheduling where both offline and online cases are considered, while distributed solution is also provided. Section~\ref{sec:exp} validates the performance of the algorithms by extensive simulations. We finally conclude the paper in Section~\ref{sec:conclusion}.

\section{Relate Work}
\label{sec:related}

With the cloud computing being used more and more widely, researchers have conducted studies on executing traditional applications (e.g., HPC and scientific computing) in cloud environments. This section summarizes the research efforts on VM assignment and scheduling that is relevant to our work in terms of operational energy management and application performance interference in data centers.

\textbf{Energy consumption management.} 
It is known that the most efficient way to reduce the energy consumption is consolidating the applications (VMs) into a set of active servers, such that the utilization of the data center is kept at a high level. An early research \cite{Nathuji07} extended virtualization solutions to support rich and effective policies for active power management which had not been done before. They integrated ``hard" and ``soft" power states to provide high power savings, and showed that substantial benefits could been derived from coordination of online methods for server consolidation with their proposed management techniques. Kusic \emph{et al.} \cite{Kusic09} considered the problem of consolidating services onto a smaller number of computing resources. They implemented a dynamic resource provisioning framework for virtualized server environments, which was tackled as one of sequential optimization and solved using a lookahead control scheme. Beloglazov \emph{et al.} \cite{Beloglazov12} investigated scheduling algorithms that consolidate VMs onto the minimum number of servers. They proposed a policy as known as Modified Best Fit Decreasing (MBFD), for energy-efficient management of cloud computing environments. There is another representative work \cite{Lin11}, in which the authors investigated the energy-saving problem by dynamically ``right-sizing" the data center in both offline and online cases. Liu \emph{et al.} \cite{Liu13} studied the problem of arbitrating the power-performance tradeoff in clouds. They provided a probabilistic framework where online decisions are made on request admission control, routing, and VM allocation.

Therefore, these works are totally different from our work as they only focus on optimizing energy consumption while guaranteeing some other metrics, such as throughput. In our approach, the performance interference is an important objective for scheduling.

\textbf{Performance interference optimization.} 
Several works  \cite{Govindan11, Chiang11, Mars11, Roytman13, Kim13, Verboven13, Xu14} have take into account the performance interference when exploiting the virtual machine consolidation to improve resource utilization. Govindan \emph{et al.} \cite{Govindan11} presented a technique for predicting performance interference due to processor cache sharing. They showed that their technique can be used to achieve the most efficient consolidation as the prediction of the performance degradation for any possible application placement only use a linear number of measurements. Chiang \emph{et al.} \cite{Chiang11} considered the problem of interference-aware scheduling for data-intensive applications in virtualized environment. They presented a task and resource allocation control framework, which can mitigate the interference effects from concurrent data-intensive applications and improve the overall system performance. Mars \emph{et al.} \cite{Mars11} presented a characterization methodology, named ``Bubble-Up", which enables the accurate prediction of the performance degradation that results from contention for shared resources in memory subsystem. They showed their methodology could predict the performance interference between collocated applications with an accuracy within $1\%$ to $2\%$ of the actual performance degradation. Roytman \emph{et al.} \cite{Roytman13} proposed a system that consolidates virtual machines to minimize the unused resources, and guarantees that the performance degradation is within a tunable bound. Their system employed a method for suitable VM combinations which was proved to perform closely to the optimal, and the system included another technique that maximizes performance while not leaving any resource unused. Kim \emph{et al.} \cite{Kim13} suggested a performance model that considers interferences in the shared last-level cache and memory bus. They claimed that the model could be used to estimate the performance degradation among applications. Based on the interference model they also presented a virtual machine consolidation method. Verboven \emph{et al.} \cite{Verboven13} addressed the performance degradation prediction models and proposed a novel approach using both the classification and regression capabilities of support vector machines. A latest survey \cite{Xu14} gave the state of the art of some of these solutions for managing the performance overhead in different cloud scenarios.

Compared with these previous works considering performance interference optimization, our model provides a unified characterization of both the energy consumption and the performance interference and our solution for VM assignment and scheduling considers both VM combination and life-cycle overlapping. We explore the tradeoff between the performance degradation overhead of VMs and resource provision of cloud data centers on a high level, which is raised as an open research issue in \cite{Esch14}. Moreover, our work can be regarded as a complement to previous works in terms of that the solutions provided by them can be integrated into our optimization framework to reduce the overall cost of a cloud system.

\section{Model and Problem Description}
\label{sec:model}
In this section we describe the model and formulate an optimization problem of VMs scheduling that aims at arbitrating between energy consumption cost and performance degradation penalty.

\subsection{Resource Allocation and Energy Cost}
We model cloud data center as an undirected graph and denote it by $G=(\mathcal{M},\mathcal{L})$, where $\mathcal{M}$ $({\vert}\mathcal{M}{\vert}=M)$ is the set of physical servers and $\mathcal{L}$ is the set of physical links between servers. Each server $server_i{\in}\mathcal{M}$ is associated with $s$ type of resources, e.g., CPU, memory, and storage space et al. The resources of $server_i$ are available in $\bm{R_i}=\{C_{ik}\}$ $(k=1,2,...,s)$ units, respectively.

Recent studies \cite{Kusic09, Fan07} have shown that the power consumption $P[u(t)]$ and the CPU utilization $u(t)$ of a server has a linear relationship
\begin{equation}
P[u(t)]=P_{idle}+(P_{peak}-P_{idle})*u(t),\label{eq:powerCalc}
\end{equation}
The $P_{idle}$ and $P_{peak}$ represent the power consumption by a server at the CPU utilization of 0\% and 100\%, respectively. Obviously, the energy consumption of a server is its power integrated over duration time, i.e., $\int_tP[u(t)]\,\mathrm{d}t$.

\subsection{Virtual Machine Request and Interference}
Cloud computing provides users with scalable, elastic and on-demand resources. Users submit their VM requests to cloud data center scheduler. Each VM request $vm_j$ is specified by an instance vector $\bm{I}_j=[a_j, p_j, \bm{R}_j]$, where $a_j$ is the arrival time, and $p_j$ is the work of processing time when $vm_j$ runs alone. Note that the VMs should start at the arrival time. The capacity vector $\bm{R}_j=\{R_{jk}\}$ $(k=1,2,...,s)$ represents the resources that $vm_j$ requires for processing its work. For example, an instance type of VM in Amazon EC2 \cite{Amazon} specifics its resource capacity \{CPU:2 vcpu/8 EC2 units, memory:7GB, storage:1680GB\}.

For each pair of $vm_j$ and $vm_{j^{'}}$, it defines the degradation factor $d_{jj^{'}}{\geq}0$ as the percentage increase in the execution time of $vm_j$ when they run concurrently on the same server. It is assumed that the performance degradation factor $d_{jj^{'}}$ between each pair of VMs, when allocating together, is known from existing methods \cite{Govindan11, Mars11, Roytman13, Verboven13, Koh07}, and we focus on the virtual machine scheduling given these factors. Note that $d_{jj^{'}}$ may not equal to $d_{j^{'}j}$ as two VMs will experience different degradation suffering from each other. It is also noted that adding VMs to the server to concurrently run with exist VMs will not reduce the degradation of previous VMs \cite{Roytman13}. It defines the degradation factor $d_{j\mathcal{J}}$ of $vm_j$ when it concurrently runs with a set $\mathcal{J}$ of VMs. Without loss of generality, it defines the $d_{j\mathcal{J}}$ as
\begin{equation}
d_{j\mathcal{J}}=\Pi_{j^{'}\in{\mathcal{J}}}(1+d_{jj^{'}})-1.
\end{equation}
This model is used to instead of $\Sigma_{j^{'}\in{\mathcal{J}}}d_{jj^{'}}$ as it is reasonable to give more severe penalty for performance degradation additive. Then the degradation factor is used to transform the processing time work. I.e., when $vm_j$ concurrently runs with a set $\mathcal{J}$ of VMs for duration time $\widetilde{p}_j$, it finishes $\frac{1}{1+d_{j\mathcal{J}}}\widetilde{p}_j$ work of processing time. To illustrate the behaviour of this interference model between VMs consider the example of Fig.~\ref{fig:example_degradationFactor}. During the first 2 unit time, $vm_1$ is collocated with $vm_2$. Each of them processes 1 unit of work because $\frac{2}{1+d_{12}}=\frac{2}{1+d_{21}}=1$. In the next 2 unit time, as the $vm_3$ joins in, all of them process 0.5 unit of work because $\frac{2}{(1+d_{12})(1+d_{13})}=\frac{2}{(1+d_{21})(1+d_{23})}=\frac{2}{(1+d_{31})(1+d_{32})}=0.5$. From the time 4 to 6, it is the same as time 2 to 4, they process 0.5 unit of work. $vm_2$ leaves the server at time 6 when it finishes its processing work. From time 6 to 8, $vm_1$ and $vm_3$ process 1 unit of work. $vm_3$ leaves at time 8 as it finishes its processing work. At last, $vm_1$ will process 1 more unit time to finish its work if it runs alone or the server is assigned VMs that do not cause performance degradation to $vm_1$.

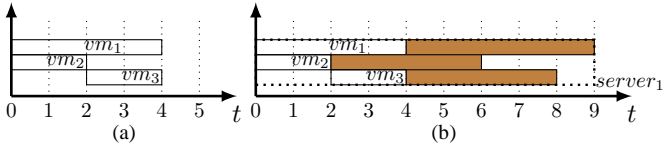
\begin{figure}
\begin{tikzpicture}
    \draw[->, >=latex,line width=1pt] (0.0,0)--(3.0,0) node[pos=1,below] {$t$};
    \draw[->, >=latex,line width=1pt] (0,0.0)--(0,1.25);
    \draw[dotted] (0.0,1.00)--(0.0,0.0) node[below, scale=0.8] {$0$};
    \draw[dotted] (0.5,1.00)--(0.5,0.0) node[below, scale=0.8] {$1$};
    \draw[dotted] (1,1.00)--(1,0.0) node[below, scale=0.8] {$2$};
    \draw[dotted] (1.5,1.00)--(1.5,0.0) node[below, scale=0.8] {$3$};
    \draw (1.5,0.0) node[below=8pt, scale=0.8] {(a)};
    \draw[dotted] (2,1.00)--(2,0.0) node[below, scale=0.8] {$4$};
    \draw[dotted] (2.5,1.00)--(2.5,0.0) node[below, scale=0.8] {${5}$};
        
    \draw(0.0,0.75) rectangle (2,0.55) node[pos=0.5,text width=0.25em, scale=0.8]{$vm_{1}$};

    \draw(0,0.55) rectangle (1.0,0.35) node[pos=0.5,text width=0.25em, scale=0.8]{$vm_{2}$};
    
    \draw(1.0,0.35) rectangle (2.0,0.15) node[pos=0.5,text width=0.25em, scale=0.8]{$vm_{3}$};
    
    \draw[->, >=latex,line width=1pt] (3.25,0)--(8.0+0.25,0) node[pos=1,below] {$t$};
    \draw[->, >=latex,line width=1pt] (3.25,0.0)--(3.25,1.25);
    \draw[dotted] (3.25,1.00)--(3.25,0.0) node[below, scale=0.8] {$0$};

    \draw[dotted] (3.75,1.00)--(3.75,0.0) node[below, scale=0.8] {$1$};

    \draw[dotted] (4.25,1.00)--(4.25,0.0) node[below, scale=0.8] {$2$};

    \draw[dotted] (4.75,1.00)--(4.75,0.0) node[below, scale=0.8] {$3$};

    \draw[dotted] (5.25,1.00)--(5.25,0.0) node[below, scale=0.8] {$4$};

    \draw[dotted] (5.75,1.00)--(5.75,0.0) node[below, scale=0.8] {$5$};

    \draw (5.75,0.0) node[below=8pt, scale=0.8] {(b)};

    \draw[dotted] (6.25,1.00)--(6.25,0.0) node[below, scale=0.8] {$6$};

    \draw[dotted] (6.75,1.00)--(6.75,0.0) node[below, scale=0.8] {${7}$};

    \draw[dotted] (7.25,1.00)--(7.25,0.0) node[below, scale=0.8] {${8}$};

    \draw[dotted] (7.75,1.00)--(7.75,0.0) node[below, scale=0.8] {${9}$};

    \draw(3.25,0.75) rectangle (5.25,0.55) node[pos=0.5,text width=0.25em, scale=0.8]{$vm_{1}$} ;
    \draw[fill=brown](5.25,0.75) rectangle (7.75,0.55);
    
    \draw(3.25,0.55) rectangle (4.25,0.35) node[pos=0.5,text width=0.25em, scale=0.8]{$vm_{2}$};
    \draw[fill=brown](4.25,0.55) rectangle (6.25,0.35);
        
    \draw(4.25,0.35) rectangle (5.25,0.15) node[pos=0.5,text width=0.25em, scale=0.8]{$vm_{3}$};
    \draw[fill=brown] (5.25,0.35) rectangle (7.25,0.15);
 
    \draw[dotted, line width=0.9pt](3.25,0.75) rectangle (7.75,0.15) node[right=0pt,pos=0.98,scale=0.8] {$server_1$}; 
\end{tikzpicture}
\caption{\label{fig:example_degradationFactor}The execution of VM Instances collocation in a server. Assume the server has resource capacity of \{$\bm{1}$\} unit. Each VM from \{$vm_1,vm_2,vm_3$\} demands \{$\bm{\frac{1}{3}}$\} unit and the processing times for the three VMs are given by $p_1=4$, $p_2=2$, and $p_3=2$. The performance degradation factors among them are all $\{d_{..}\}=1$. (a): VM Instance configuration. (b): Running in a server.}
\end{figure}

\subsection{Scheduling Problem Description}

There are two issues need to be concerned about the allocation of virtual machines. Cloud infrastructure providers offer some specific kinds of VMs, which tend to reserve resources, such as CPU, memory and storage space. They pursue to reduce the operational cost, i.e., minimize energy consumption or capacity cost. On the other side, cloud users seek to reduce the running time of their requests. In this way, they can save bills for the rented resources.

In our scheduling, time is divided into discrete periods, $t=1,2,...,T$. For example, the interval $\tau$ can be one or five minute(s). The binary decision variable $x_{ij}(t)$ indicates whether $vm_j$ is allocated to $server_i$ at the time slot $t$. It defines $c_i(t)$ as the energy consumption cost of $server_i$ running during time slot $t$, i.e., $c_i(t)=P[u(t)]\tau$. Thus, the total operational cost during time slot $t$ is the sum of all running servers, which is calculated as
\begin{equation}
C_i(t)=\sum_{i}c_i(t),
\end{equation}

Let $Q(t)$ denote the set of VMs that run at time slot $t$. Define $D(t)$ as the set of VMs that complete their execution and leave at the time slot $t$. Hence, the execution time $t_j$ of $vm_j$ $(j\in{D(t)})$ is
\begin{equation}
t_j=t-a_j, t_j{\geq}p_j.
\end{equation}

The performance degradation penalty is model by a convex function $f(\cdot)$. One natural model for it is $f[(\frac{t_j-p_j}{p_j})^+]=\alpha^{(\frac{t_j-p_j}{p_j})^+}-1$ $([x]^+{\equiv}max\{0, x\})$, which penalizes the delay cost from the processing time $p_j$. Therefore, the VMs scheduling problem is defined as the following optimization:
\begin{equation}
min \qquad \sum\limits_{t=1}^T\sum\limits_{i=1}^{M}c_i(t)+\beta \sum\limits_{t=1}^{T}\sum_{j{\in}D(t)}f[(\frac{t_j-p_j}{p_j})^+]
\end{equation}
\begin{eqnarray}
s.t.\quad \sum_{j{\in}Q(t)}x_{ij}(t)R_{jk}{\leq}R_{ik} \quad \forall{i},\forall{t},\forall{k},\\
\sum_{i=1}^Mx_{ij}(t)=1 \quad \forall{j},\forall{t},\\
x_{ij}(t){\leq}x_{ij}(t+1) \quad \forall{t},j{\notin}D(t),\\
x_{ij}(t){\in}\{0,1\} \quad \forall{i},\forall{j},\forall{t}.
\end{eqnarray}
The objective function $(5)$ minimizes the total operational server costs and performance degradation penalties, and $\beta>0$ is some constant and represents the relative importance between two objectives. Constraint $(6)$ ensures that the aggregated resource demand of multiple VMs does not exceed a server's capacity for all resource types and at all time slot. Constraint $(7)$ relates to that each VM is allocated to one of the servers at any point in time. Constraint $(8)$ refers to that if a VM has assigned on a server it will not be assigned to other servers. Constraint $(9)$ follows that $x_{ij}(t)$ is set to one if $vm_j$ is allocated to $server_i$ at time slot $t$. We first give the computational complexity of this problem as following:
\newtheorem{theorem}{Theorem}
\begin{theorem}
Find an optimal VMs schedule for arbitrating between operational cost and performance degradation penalty is NP-Complete. \label{theorem:hardness}
\end{theorem}
\begin{proof}
First, we transform the optimization problem to an associated decision problem: given the instance vectors of VMs, the performance degradation factors, and a bound on the sum of energy consumption and performance degradation penalty, is there a schedule such that the bound on sum of cost and penalty is satisfied? Clearly, it belongs to NP, since we can computing and verify in polynomial time that a proposed schedule satisfies the given bound on the sum of operational cost and performance degradation penalty. We next prove that finding an optimal VMs schedule for arbitrating between energy consumption and performance degradation penalty is NP-Complete via the reduction to the 3-Dimensional Matching problem \cite{Garey79, Leung04}.

Consider an instance of 3-Dimensional Matching: Let $A=\{a_1,a_2,..,a_q\}$, $B=\{b_1,b_2,...,b_q\}$, and $C=\{c_1,c_2,...,c_q\}$ be three disjoint sets of $q$ elements each. Let $Z=\{z_1,z_2,...,z_{l^{'}},...,z_{l}\}$ be a set of triples such that each $z_i$ consists of one element from $A$, one element from $B$, and one element from $C$. Is there a subset $Z^{'}{\subseteq}Z$ such that every element in $A$, $B$, and $C$ appears in exactly one triple in $Z^{'}$? We construct an instance of VM scheduling problem as follows. Let there be $3q$ VMs and $M_q({\geq}q)$ servers. The VMs correspond to the elements in $A$, $B$ and $C$. For each $1{\leq}i{\leq}M_q$, $server_i$ has resource vector $\bm{R_i}=\{R_{ik}\}=\{\bm{1}\}$ $(k=1,2,...,s)$. For each $1{\leq}j{\leq}3q$, $vm_j$ has instance vector $[0,1,\{\bm{\frac{1}{3}}\}]$. VMs have no interference with each other in the triples $z_{l^{'}}{\in}Z^{'}$; otherwise, they have performance degradation factor 1 between each other. The sum of cost and penalty is $q$. The energy consumption cost of a server is $1$ per unit of time slot when it runs at full utilization (Suppose it be $0.5$ at idle). The sum of cost and penalty is equal to $q$ if and only if the $3q$ VMs are scheduled on $q$ servers and do not cause performance degradation. I.e., $\bigcup\limits_{server_i}z_{l^{'}}=A{\cup}B{\cup}C$. Thus, there is an optimal VMs schedule if and only if there is a 3-Dimensional matching. It is clear that the above reduction is a pseudo-polynomial reduction. So we can conclude that the problem is NP-Complete by this pseudo-polynomial time reduction to the 3-Dimensional Matching problem which has been proved to be NP-Complete.
\end{proof}

\section{Virtual Machine Scheduling Design}
\label{sec:algorithm}

As it is a NP-hard combinatorial optimization problem and there is no computationally-efficient solution, we exploit the unique problem structure of VM scheduling in cloud data centers to develop the solutions. We first study a static problem (offline). After then, we develop the solution to the dynamic version of the problem (online).

\subsection{Offline scheduling Design}
In this condition, the informations of VMs that will be scheduled are known at the outset. We propose offline algorithms for virtual machine scheduling and analyse the performance.

\textbf{Bin Packing Variant Algorithm (BPV).} From the perspective of single energy consumption criterion optimization, various packing algorithms are become the reserve choices. It is an obvious advantage to reduce the energy consumption when decreasing the number of active servers. So an algorithm derived from First-Fit bin packing is considered. The algorithm keeps the VMs in a list sorted in increasing order of the arrival time. Each VM is allocated to the first possible accommodated server according to the list order. It invokes a server when capacity violation happens. The difference from First Fit algorithm is VMs will depart from the servers when they finish their work and the relevant resources will be recovered.

\textbf{Minimum Increasing Cost Algorithm (MIC).} Another natural algorithm is greedy differential of increasing costs of energy consumption and performance degradation penalty. The VMs are also kept in the increasing order of their arrival time. It would assign the next VM to the server that minimizes the increment of total cost. There would be two choices for the allocation of next VM. The increment of total cost is the sum of energy consumption and performance degradation penalty when the VM is allocated to an active server running with exist VMs. The other choice is a currently unused server with paying for more static energy consumption that supposing the VM process alone.
\newtheorem{lemma}{Lemma}
\begin{theorem}
Let $I_{max}$ denote the maximum number of VMs that can be simultaneously accommodated by a server. The approximation ratio of MIC algorithm is $I_{max}$. \label{theorem:approximation}
\end{theorem}
\begin{proof}
Note that for the minimization problem, an algorithm achieves a $\delta$-approximation factor if for all instances it returns a solution at most $\delta$ times the optimal value.

We decompose the power of a $server_i$ at time $t$ into the VMs according to the proportion of their CPU resources. For example, at time $t$, there are $n_t$ VMs with CPU resource of $R_{j1}$ $(j=1,...,n_t)$ (here we set the resource type $1$ as the CPU resource) in $server_i$ which has CPU resource $C_{i1}$. Then $vm_j$ consumes $\frac{R_{j1}P_{idle}}{\sum_{j=1}^{n_t}R_{j1}}+\frac{R_{j1}(P_{peak}-P_{idle})}{C_{i1}}$ power respectively, where the former part corresponds to the proportion of the static power and the later part is the dynamic power this VM consumes. Without loss of generality, We consider the VM $vm_j$. The power of this VM during its execution time is at least $\frac{R_{j1}P_{idle}}{C_{i1}}+\frac{R_{j1}(P_{peak}-P_{idle})}{C_{i1}}$ because $\sum_{j=1}^{n_t}R_{j1}{\leq}C_{i1}$. Again, we decompose the total cost of energy consumption and performance degradation penalty to the cost of each VM when it is allocated. According to the MIC algorithm, the cost of inserted $vm_j$ is no more than $[P_{idle}+\frac{R_{j1}(P_{peak}-P_{idle})}{C_{j1}}]*p_j$. So the approximation ratio is
\begin{equation*}
\begin{split}
\frac{cost(MIC)}{cost(OPT)}&{{\leq}_{\textcircled{1}}}\frac{\sum\limits_{j}[P_{idle}+\frac{R_{j1}(P_{peak}-P_{idle})}{C_{j1}}]*p_j}{\sum\limits_{j}[\frac{R_{j1}P_{idle}}{C_{i1}}+\frac{R_{j1}(P_{peak}-P_{idle})}{C_{i1}}]*p_j}\\
&{=}\frac{\sum\limits_{j}P_{idle}*p_j+\sum\limits_{j}\frac{R_{j1}(P_{peak}-P_{idle})}{C_{j1}}*p_j}{\sum\limits_{j}\frac{R_{j1}P_{idle}}{C_{i1}}*p_j+\sum\limits_{j}\frac{R_{j1}(P_{peak}-P_{idle})}{C_{i1}}*p_j}\\
&{{\leq}_{\textcircled{2}}}\frac{\sum\limits_{j}P_{idle}*p_j}{\sum\limits_{j}\frac{R_{j1}P_{idle}}{C_{i1}}*p_j}{{\leq}_{\textcircled{3}}}\frac{\sum\limits_{j}P_{idle}*p_j}{\sum\limits_{j}\frac{P_{idle}}{I_{max}}*p_j}={I_{max}}
\end{split}
\end{equation*}
where the second inequality follows from $\sum\limits_{j}P_{idle}*p_j{\geq}\sum\limits_{j}\frac{R_{j1}P_{idle}}{C_{i1}}*p_j$ as $R_{j1}{\leq}C_{i1}$, and a mathematical inequality $\frac{a+c}{b+c}{\leq}\frac{a}{b}$ as $a{\geq}b, c{\geq}0$. The third inequality results from $I_{max}*R_{j1}{\geq}C_{i1}$, i.e., $\frac{R_{j1}P_{idle}}{C_{i1}}{\geq}\frac{P_{idle}}{I_{max}}$. This concludes the Theorem.
\end{proof}
Remark: BPV algorithm only considers to accept the next VM, and does not take into account the performance degradation. Both of the above algorithms sort the VMs by their arrival time and depend only on the information that is available to the algorithms at the scheduling time. So they are also online algorithms. Note that when a VM is allocated to a server, there is a need to update duration time of itself and other VMs that are interfered by it.

An observation is that these algorithms do not consider the life cycle overlapping of VMs. For example, there are three VMs to be scheduled, which are configured as Fig.~\ref{fig:example_overlapping}. In Fig.~\ref{fig:example_overlapping}(a), both of above algorithms cause 2 servers to be active from time 0 to 11 and 1 to 11. A better scheduling (Fig.~\ref{fig:example_overlapping}(b)) is that it assigns $vm_1$  and $vm_3$ in $server_1$ and $vm_2$ in $server_2$. Then we can put $server_2$ into power-saving mode or turn-off from time 2 to 11.
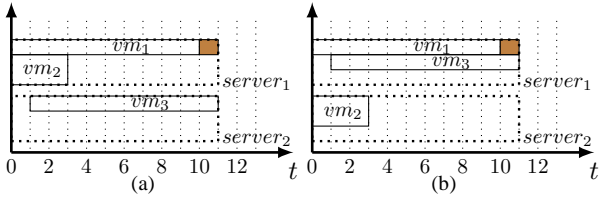
\begin{figure}
\begin{tikzpicture}
    \draw[->, >=latex,line width=1pt] (0.0,0)--(3.75,0) node[pos=1,below] {$t$};
    \draw[->, >=latex,line width=1pt] (0,0.0)--(0,2);
    \draw[dotted] (0.0,1.75)--(0.0,0.0) node[below, scale=0.8] {$0$};
    \draw[dotted] (0.25,1.75)--(0.25,0.0);
    \draw[dotted] (0.5,1.75)--(0.5,0.0) node[below, scale=0.8] {$2$};
    \draw[dotted] (0.75,1.75)--(0.75,0.0);
    \draw[dotted] (1,1.75)--(1,0.0) node[below, scale=0.8] {$4$};
    \draw[dotted] (1.25,1.75)--(1.25,0.0);
    \draw[dotted] (1.5,1.75)--(1.5,0.0) node[below, scale=0.8] {$6$};
    \draw[dotted] (1.75,1.75)--(1.75,0.0) node[below=6pt, scale=0.8] {(a)};
    \draw[dotted] (2,1.75)--(2,0.0) node[below, scale=0.8] {$8$};
    \draw[dotted] (2.25,1.75)--(2.25,0.0);
    \draw[dotted] (2.5,1.75)--(2.5,0.0) node[below, scale=0.8] {${10}$};
    \draw[dotted] (2.75,1.75)--(2.75,0.0);
    \draw[dotted] (3,1.75)--(3,0.0) node[below, scale=0.8] {${12}$};
    \draw[dotted] (3.25,1.75)--(3.25,0.0);

    \draw(0.0,1.5) rectangle (2.5,1.3) node[pos=0.55,text width=0.25em, scale=0.8] {$vm_{1}$};
    \draw[fill=brown](2.5,1.5) rectangle (2.75,1.3);

    \draw(0.0,1.3) rectangle (0.75,0.9) node[pos=0.55, scale=0.8]{$vm_{2}$};
        
    \draw(0.25,0.75) rectangle (2.75,0.55) node[pos=0.55,text width=0.25em, scale=0.8]{$vm_{3}$};
    
    \draw[dotted,line width=0.9pt](0.0,1.5) rectangle (2.75,0.9) node[right=0pt,pos=0.98,scale=0.8] {$server_1$};    
    \draw[dotted,line width=0.9pt](0.0,0.75) rectangle (2.75,0.15) node[right=0pt,pos=0.98,scale=0.8] {$server_2$}; 
        
    \draw[->, >=latex,line width=1pt] (4.0,0)--(7.75,0) node[pos=1,below] {$t$};
    \draw[->, >=latex,line width=1pt] (4,0.0)--(4,2);
    \draw[dotted] (4.0,1.75)--(4.0,0.0) node[below, scale=0.8] {$0$};
    \draw[dotted] (4.25,1.75)--(4.25,0.0);
    \draw[dotted] (4.5,1.75)--(4.5,0.0) node[below, scale=0.8] {$2$};
    \draw[dotted] (4.75,1.75)--(4.75,0.0);
    \draw[dotted] (5,1.75)--(5,0.0) node[below, scale=0.8] {$4$};
    \draw[dotted] (5.25,1.75)--(5.25,0.0);
    \draw[dotted] (5.5,1.75)--(5.5,0.0) node[below, scale=0.8] {$6$};
    \draw[dotted] (5.75,1.75)--(5.75,0.0) node[below=6pt, scale=0.8] {(b)};
    \draw[dotted] (6,1.75)--(6,0.0) node[below, scale=0.8] {$8$};
    \draw[dotted] (6.25,1.75)--(6.25,0.0);
    \draw[dotted] (6.5,1.75)--(6.5,0.0) node[below, scale=0.8] {${10}$};
    \draw[dotted] (6.75,1.75)--(6.75,0.0);
    \draw[dotted] (7,1.75)--(7,0.0) node[below, scale=0.8] {${12}$};
    \draw[dotted] (7.25,1.75)--(7.25,0.0);  
    
    \draw(4.0,1.5) rectangle (6.5,1.3) node[pos=0.55,text width=0.25em, scale=0.8] {$vm_{1}$};
    \draw[fill=brown](6.5,1.5) rectangle (6.75,1.3);
    
    \draw(4.25,1.3) rectangle (6.75,1.1) node[pos=0.55,text width=0.25em, scale=0.8]{$vm_{3}$};
        
    \draw(4.0,0.75) rectangle (4.75,0.35) node[pos=0.55, scale=0.8]{$vm_{2}$};

    \draw[dotted,line width=0.9pt](4.0,1.5) rectangle (6.75,0.9) node[right=0pt,pos=0.98,scale=0.8] {$server_1$};    
    \draw[dotted,line width=0.9pt](4.0,0.75) rectangle (6.75,0.15) node[right=0pt,pos=0.98,scale=0.8] {$server_2$}; 
 
\end{tikzpicture}
\caption{\label{fig:example_overlapping}BPV+MIC and an improved schedule for three VMs. Assume both physical servers have resource capacity of \{$\bm{1}$\} unit and $vm_1$, $vm_2$ and $vm_3$ have resource capacity of \{$\bm{\frac{1}{3}}$\}, \{$\bm{\frac{2}{3}}$\}, and \{$\bm{\frac{1}{3}}$\} unit respectively. The performance degradation factors among them are \{$d_{12} = d_{13} = 0$, $d_{21} = d_{31} = 0.1$, $d_{23} = d_{32} = 0.2$\} while the processing times for the three VMs are given by $p_1 = 10$, $p_2 = 3$, and $p_3 = 9$. (a): BPV or MIC scheduling. (b): A better scheduling.}
\end{figure}

\textbf{Maximum Decreasing Cost Algorithm (MDC)}. Instead of sorting the VMs by the arrival time, MDC algorithm considers the information of all VMs and works like the clustering algorithm. We pursue the minimum cost of energy consumption and performance degradation penalty iteratively. Initially, each VM is allocated to a dedicated server. Next, we decide to repeatedly merge servers together by the form of pairs. The process of merging is to collocate the VMs on one server. There is also a need to update duration time of the VMs which cause interference among them. We define the gain function of merging two server, $\mathcal{S}_u$ and $\mathcal{S}_v$ as the following:
\begin{equation}
Gain(\mathcal{S}_u,\mathcal{S}_v)=Cost(\mathcal{S}_u)+Cost(\mathcal{S}_v)-Cost(\mathcal{S}_u{\cup}\mathcal{S}_v),
\end{equation}
where $Cost(\cdot)$ denotes the total cost of the server according to the cost model defined in Section~\ref{sec:model}-C. With regard to the merger which causes the violation of server capacity, we define the gain as negative number. At each step we choose the merger of two servers that results in the maximum decrease in the total cost. The algorithm ends when the merger of any two servers will produce an negative gain. The pseudo-code for MDC algorithm is summarized in \textbf{\textit{Algorithm.~\ref{algo:MDC_algorithm}}}.
\begin{algorithm}
\SetKwInOut{Input}{input}\SetKwInOut{Output}{output}
\Input{the set of VMs $\mathcal{N}$}
\Output{the scheduling result of VMs}
\Begin{
Initial Servers $\mathcal{S}$=$\{\mathcal{S}_1=\{vm_1\},...,\mathcal{S}_n=\{vm_n\}\}$\; $maxGain=\max_{u,v}Gain(\mathcal{S}_u,\mathcal{S}_{v})$\;
\While{$maxGain{\geq}0$}{
MergeServers($\mathcal{S}_u,\mathcal{S}_v$)\;
Set $\mathcal{S}_u=\mathcal{S}_u{\cup}\mathcal{S}_v$, $\mathcal{S}=\mathcal{S}{\setminus}\mathcal{S}_v$\;
$maxGain=\max_{u,v}Gain(\mathcal{S}_u,\mathcal{S}_{v})$\;
}
Return the set of servers $\mathcal{S}$ and their accommodated VMs, correspondingly.
}
\caption{\label{algo:MDC_algorithm}Maximum Decreasing Cost Algorithm}
\end{algorithm}

To illustrate the different behaviour of these three scheduling strategies we present an example in Fig.~\ref{fig:behaviour_algo}. BPV aggregates the VMs in parts of servers and leaves some servers to be low utilization. MIC is more likely to assign the subsequent VMs to be included by the anterior ones duration its execution when their performance degradation factor is low. As a result, it considers to balance the VMs between the servers. MDC prefers to collocate the VMs that share long life cycle and have low performance degradation factors between them. In summary, these scheduling algorithms have different performances and we will evaluate them in Section~\ref{sec:exp}.
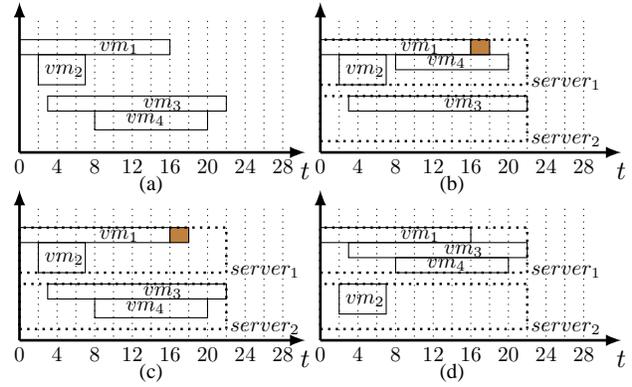
\begin{figure}
\begin{tikzpicture}
    \draw[->, >=latex,line width=1pt] (0.0,0)--(3.75+0.05,0) node[pos=1,below] {$t$};
    \draw[->, >=latex,line width=1pt] (0,0.0)--(0,2);
    \draw[dotted] (0.0,1.75)--(0.0,0.0) node[below, scale=0.8] {$0$};
    \draw[dotted] (0.25,1.75)--(0.25,0.0);
    \draw[dotted] (0.5,1.75)--(0.5,0.0) node[below, scale=0.8] {$4$};
    \draw[dotted] (0.75,1.75)--(0.75,0.0);
    \draw[dotted] (1,1.75)--(1,0.0) node[below, scale=0.8] {$8$};
    \draw[dotted] (1.25,1.75)--(1.25,0.0);
    \draw[dotted] (1.5,1.75)--(1.5,0.0) node[below, scale=0.8] {$12$};

    \draw[dotted] (1.75,1.75)--(1.75,0.0) node[below=6pt, scale=0.8] {(a)};
    \draw[dotted] (2,1.75)--(2,0.0) node[below, scale=0.8] {$16$};
    \draw[dotted] (2.25,1.75)--(2.25,0.0);
    \draw[dotted] (2.5,1.75)--(2.5,0.0) node[below, scale=0.8] {${20}$};
    \draw[dotted] (2.75,1.75)--(2.75,0.0);
    \draw[dotted] (3,1.75)--(3,0.0) node[below, scale=0.8] {${24}$};
    \draw[dotted] (3.25,1.75)--(3.25,0.0);
    \draw[dotted] (3.5,1.75)--(3.5,0) node[below, scale=0.8] {${28}$}; 
    
    \draw(0.0,1.5) rectangle (2.0,1.3) node[pos=0.55,text width=0.25em, scale=0.8] {$vm_{1}$};   
    
    \draw(0.25,1.3) rectangle (0.875,0.9) node[pos=0.55, scale=0.8]{$vm_{2}$};
        
    \draw(0.375,0.75) rectangle (2.75,0.55) node[pos=0.55,text width=0.25em, scale=0.8]{$vm_{3}$};

    \draw(1.0,0.55) rectangle (2.5,0.30) node[pos=0.45,scale=0.8]{$vm_{4}$};

    \draw[->, >=latex,line width=1pt] (4.0,0)--(7.75+0.25,0) node[pos=1,below] {$t$};
    \draw[->, >=latex,line width=1pt] (4,0.0)--(4,2);
    \draw[dotted] (4.0,1.75)--(4.0,0.0) node[below, scale=0.8] {$0$};
    \draw[dotted] (4.25,1.75)--(4.25,0.0);
    \draw[dotted] (4.5,1.75)--(4.5,0.0) node[below, scale=0.8] {$4$};
    \draw[dotted] (4.75,1.75)--(4.75,0.0);
    \draw[dotted] (5,1.75)--(5,0.0) node[below, scale=0.8] {$8$};
    \draw[dotted] (5.25,1.75)--(5.25,0.0);
    \draw[dotted] (5.5,1.75)--(5.5,0.0) node[below, scale=0.8] {$12$};

    \draw[dotted] (5.75,1.75)--(5.75,0.0) node[below=6pt, scale=0.8] {(b)};
    \draw[dotted] (6,1.75)--(6,0.0) node[below, scale=0.8] {$16$};
    \draw[dotted] (6.25,1.75)--(6.25,0.0);
    \draw[dotted] (6.5,1.75)--(6.5,0.0) node[below, scale=0.8] {${20}$};
    \draw[dotted] (6.75,1.75)--(6.75,0.0);
    \draw[dotted] (7,1.75)--(7,0.0) node[below, scale=0.8] {${24}$};
    \draw[dotted] (7.25,1.75)--(7.25,0.0);
    \draw[dotted] (7.5,1.75)--(7.5,0.0) node[below, scale=0.8] {${28}$};  
    
    \draw(4.0,1.5) rectangle (6,1.3) node[pos=0.55,text width=0.25em, scale=0.8] {$vm_{1}$};
    \draw[fill=brown](6,1.5) rectangle (6.25,1.3);
    
    \draw(4.25,1.3) rectangle (4.875,0.9) node[pos=0.55, scale=0.8]{$vm_{2}$};
    
    \draw(5.0,1.3) rectangle (6.5,1.1) node[pos=0.45,scale=0.8]{$vm_{4}$};
    
    \draw(4.375,0.75) rectangle (6.75,0.55) node[pos=0.55,text width=0.25em, scale=0.8]{$vm_{3}$};
           
    \draw[dotted,line width=0.9pt](4.0,1.5) rectangle (6.75,0.9) node[right=0pt,pos=0.98,scale=0.8] {$server_1$};    
    \draw[dotted,line width=0.9pt](4.0,0.75) rectangle (6.75,0.15) node[right=0pt,pos=0.98,scale=0.8] {$server_2$};

    \draw[->, >=latex,line width=1pt] (0.0,-2.5)--(3.75+0.05,-2.5) node[pos=1,below] {$t$};
    \draw[->, >=latex,line width=1pt] (0,-2.5)--(0,-0.5);
    
    \draw[dotted] (0.0,1.75-2.5)--(0.0,0.0-2.5) node[below, scale=0.8] {$0$};
    \draw[dotted] (0.25,1.75-2.5)--(0.25,0.0-2.5);
    \draw[dotted] (0.5,1.75-2.5)--(0.5,0.0-2.5) node[below, scale=0.8] {$4$};
    \draw[dotted] (0.75,1.75-2.5)--(0.75,0.0-2.5);
    \draw[dotted] (1,1.75-2.5)--(1,0.0-2.5) node[below, scale=0.8] {$8$};
    \draw[dotted] (1.25,1.75-2.5)--(1.25,0.0-2.5);
    \draw[dotted] (1.5,1.75-2.5)--(1.5,0.0-2.5) node[below, scale=0.8] {$12$};

    \draw[dotted] (1.75,1.75-2.5)--(1.75,0.0-2.5) node[below=6pt, scale=0.8] {(c)};
    \draw[dotted] (2,1.75-2.5)--(2,0.0-2.5) node[below, scale=0.8] {$16$};
    \draw[dotted] (2.25,1.75-2.5)--(2.25,0.0-2.5);
    \draw[dotted] (2.5,1.75-2.5)--(2.5,0.0-2.5) node[below, scale=0.8] {${20}$};
    \draw[dotted] (2.75,1.75-2.5)--(2.75,0.0-2.5);
    \draw[dotted] (3,1.75-2.5)--(3,0.0-2.5) node[below, scale=0.8] {${24}$};
    \draw[dotted] (3.25,1.75-2.5)--(3.25,0.0-2.5);
    \draw[dotted] (3.5,1.75-2.5)--(3.5,0-2.5) node[below, scale=0.8] {${28}$}; 
    
    \draw(0.0,1.5-2.5) rectangle (2,1.3-2.5) node[pos=0.55,text width=0.25em, scale=0.8] {$vm_{1}$};
    \draw[fill=brown](2,1.5-2.5) rectangle (2.25,1.3-2.5);
    
    \draw(0.25,1.3-2.5) rectangle (0.875,0.9-2.5) node[pos=0.55, scale=0.8]{$vm_{2}$};
        
    \draw(0.375,0.75-2.5) rectangle (2.75,0.55-2.5) node[pos=0.55,text width=0.25em, scale=0.8]{$vm_{3}$};
    
    \draw(1.0,0.55-2.5) rectangle (2.5,0.30-2.5) node[pos=0.45,scale=0.8]{$vm_{4}$};
    
    \draw[dotted,line width=0.9pt](0.0,1.5-2.5) rectangle (2.75,0.9-2.5) node[right=0pt,pos=0.98,scale=0.8] {$server_1$};    
    \draw[dotted,line width=0.9pt](0.0,0.75-2.5) rectangle (2.75,0.15-2.5) node[right=0pt,pos=0.98,scale=0.8] {$server_2$};

    \draw[->, >=latex,line width=1pt] (4.0,0-2.5)--(7.75+0.25,0-2.5) node[pos=1,below] {$t$};
    \draw[->, >=latex,line width=1pt] (4,0.0-2.5)--(4,2-2.5);
    \draw[dotted] (4.0,1.75-2.5)--(4.0,0.0-2.5) node[below, scale=0.8] {$0$};
    \draw[dotted] (4.25,1.75-2.5)--(4.25,0.0-2.5);
    \draw[dotted] (4.5,1.75-2.5)--(4.5,0.0-2.5) node[below, scale=0.8] {$4$};
    \draw[dotted] (4.75,1.75-2.5)--(4.75,0.0-2.5);
    \draw[dotted] (5,1.75-2.5)--(5,0.0-2.5) node[below, scale=0.8] {$8$};
    \draw[dotted] (5.25,1.75-2.5)--(5.25,0.0-2.5);
    \draw[dotted] (5.5,1.75-2.5)--(5.5,0.0-2.5) node[below, scale=0.8] {$12$};
 
    \draw[dotted] (5.75,1.75-2.5)--(5.75,0.0-2.5) node[below=6pt, scale=0.8] {(d)};
    \draw[dotted] (6,1.75-2.5)--(6,0.0-2.5) node[below, scale=0.8] {$16$};
    \draw[dotted] (6.25,1.75-2.5)--(6.25,0.0-2.5);
    \draw[dotted] (6.5,1.75-2.5)--(6.5,0.0-2.5) node[below, scale=0.8] {${20}$};
    \draw[dotted] (6.75,1.75-2.5)--(6.75,0.0-2.5);
    \draw[dotted] (7,1.75-2.5)--(7,0.0-2.5) node[below, scale=0.8] {${24}$};
    \draw[dotted] (7.25,1.75-2.5)--(7.25,0.0-2.5);
    \draw[dotted] (7.5,1.75-2.5)--(7.5,0.0-2.5) node[below, scale=0.8] {${28}$};  
    
    \draw(4.0,1.5-2.5) rectangle (6,1.3-2.5) node[pos=0.55,text width=0.25em, scale=0.8] {$vm_{1}$};
    
    \draw(4.375,1.3-2.5) rectangle (6.75,1.1-2.5) node[pos=0.55,text width=0.25em, scale=0.8]{$vm_{3}$};
    
    \draw(5.0,1.1-2.5) rectangle (6.5,0.9-2.5) node[pos=0.45,scale=0.8]{$vm_{4}$};
    
    \draw(4.25,0.75-2.5) rectangle (4.875,0.35-2.5) node[pos=0.55, scale=0.8]{$vm_{2}$};
    
    \draw[dotted,line width=0.9pt](4.0,1.5-2.5) rectangle (6.75,0.9-2.5) node[right=0pt,pos=0.98,scale=0.8] {$server_1$};    
    \draw[dotted,line width=0.9pt](4.0,0.75-2.5) rectangle (6.75,0.15-2.5) node[right=0pt,pos=0.98,scale=0.8] {$server_2$}; 
    
\end{tikzpicture}
\caption{\label{fig:behaviour_algo}Four VMs scheduling. Assume each physical server has resource capacity of \{$\bm{1}$\} unit. Each VM from \{$vm_1,vm_3,vm_4$\} demands \{$\bm{\frac{1}{3}}$\} unit while $vm_2$ requires \{$\bm{\frac{2}{3}}$\} unit. The processing times for them are given by $p_1=16$, $p_2=5$, $p_3=19$, and $p_4=12$. The performance degradation factor $d_{21}$ is $0.25$, and other factors among them are $0$. (a): VM instance. (b): BPV scheduling. (c): MIC scheduling. (d): MDC scheduling.}
\end{figure}

\subsection{Online Algorithm for Dynamic Problem}
In this section, we introduce the online version of the VM scheduling, in which a sequence $\mathcal{N}$ of VMs arrive over time, where $\mathcal{N}=\{vm_0,vm_1,...,vm_j,...\}$. Each VM $vm_j$ must be assigned upon its arrival, without information about future VMs $\{vm_{j^{'}}|j^{'}>j\}$. We explore algorithm that schedules each incoming VM by dispatching them to current active servers or a new server that be activated. Recall that BPV and MIC algorithms depend only on the informations that are available to the algorithms when we schedule upon the arrival of a VM $vm_j$ at time $t$ instead of sorting the VMs by their arrival time. They can be modified to support the scheduling in online version, and they are denoted by OBPV and OMIC, respectively. Specially, we derive a competitive ratio for the OMIC algorithm. We say that for the minimization problem, an algorithm is $\gamma$-$competitive$ if for all the problem instances, it returns the cost at most $\gamma$ times the cost of the optimal offline solution. From the Theorem~\ref{theorem:approximation}, we have the following theorem:
\begin{theorem}
The competitive ratio of OMIC algorithm for VM scheduling which aims at arbitrating between operational cost and performance degradation is at most $I_{max}$, where $I_{max}$ has the same mean of Theorem~\ref{theorem:approximation}. \label{theorem:competitive}
\end{theorem}

\subsection{Incorporating VM batch arrival and VM reservation}
In the previous section, the VMs are consider to arrive one by one and there are no information about the future arriving VMs. In order to match the cloud data centers, we incorporate the following two properties to the scheduling design:
\begin{itemize}
\item There are a set $\mathcal{N}_t$ of VMs to be scheduled at time $t$ due to the many users submit their VMs to the cloud data center at the same time.
\item There are a set $\mathcal{N}_t^f$ of reserved VMs at time $t$ due to users reserve for lower costs and reserving capacity in the cloud data center.  
\end{itemize}
It defines the time $t$ as scheduling time only when there are some VMs need to be started at this time. According to the definition of the reserved VMs, We have the arrival time relationship $a_j>a_j^{'}, \forall{vm_j}{\in}\mathcal{N}_t^f,\forall{vm_j^{'}}{\in}\mathcal{N}_t$. In this situation, the problem is transformed to schedule a set of VMs, i.e., $\mathcal{N}_t{\cup}\mathcal{N}_t^f$, to be allocated on the cloud data center. The difference from the offline scheduling is that at the begin of the scheduling there are some VMs had been allocated on the cloud data center. Consider an example scheduling time $t=2$ in previous Fig.~\ref{fig:behaviour_algo}, incoming VM $vm_2$ and reserved VMs $\{vm_3,vm_4\}$ need to be scheduled at this time. We prefer the scheduling of Fig.~\ref{fig:behaviour_algo}(d) to Fig.~\ref{fig:behaviour_algo}(c) as it is known that a server can be put into power-saving mode or shut down only if there are no VMs active on it. Then it causes a problem of which VM first to be scheduled if the OMIC algorithm is used to schedule VM one by one. It defines the alignment ratio of VM to server as following: the ratio of VMs{'} completion time to server's completion time. An obvious intuition is to allocate the VMs to maximize the alignment ratio if these VMs have weak performance interference, i.e., to align the VMs and their server. Then the points is to allocate the VMs to their best candidate server. Based on the above analysis, we present an algorithm from the servers' perspective. First, it supposes that each VM $vm_j{\in}(\mathcal{N}_t{\cup}\mathcal{N}_t^f)$ is allocated on a dedicate virtual server. Each server proposes the profits to other VMs that are allocated on virtual servers. To be specific, it defines the profit metric as the following:
\begin{equation}
Profit_{i,j}=Cost(vm_j)-AddCost(server_i,vm_j),\label{eq:profitCalc}
\end{equation}
where $Cost(vm_j)$ is the total cost of a server to run $vm_j$ alone, and $AddCost(server_i,vm_j)$ is the increment total cost of run $vm_j$ on $server_i$. When $vm_j$ cannot be allocated on the server $server_i$, the $Profit_{i,j}$ is simply set to negative value. The algorithm allocates the VMs iteratively. In each round, it picks the maximum profit, i.e., $\{\max\limits_{i,j}Profit_{i,j}|Profit_{i,j}>=0\}$. Then the $vm_j$ is allocated on $server_i$. The algorithm stops when $\max\limits_{i,j}Profit_{i,j}<0$. In this situation, it says that the total cost cannot be improved. The VMs are allocated to their current servers. It is summarized in \textbf{\textit{Algorithm.~\ref{algo:IVP_algorithm}}} ($IVP$). 
\begin{algorithm}
\SetKwInOut{Input}{input}\SetKwInOut{Output}{output}
\Input{the set of VMs $\{\mathcal{N}_t{\cup}\mathcal{N}_t^f\}$, current active servers $\mathcal{S}_I$ at time $t$}
\Output{the scheduling result of VMs}
\Begin{
Set virtual servers $\mathcal{S}_{V}$=$\{\mathcal{S}_j$=$\{vm_j\}|vm_j{\in}\{\mathcal{N}_t{\cup}\mathcal{N}_t^f\}\}$\;
Set all servers $\mathcal{S}=\mathcal{S}_{V}{\cup}\mathcal{S}_{I}$\;
Set $maxProfit=\max\limits_{i,j}Profit_{i,j}$, $i{\in}\mathcal{S}$\;
\While{$maxProfit{\geq}0$}{
The configuration $Server_{i}(t)=\argmax\limits_{i{\in}\mathcal{S}}Profit_{i,j}$\;
Update $server_i$ and Delete $vm_j$ from $\{\mathcal{N}_t{\cup}\mathcal{N}_t^f\}$\;
Set $maxProfit=\max\limits_{i,j}Profit_{i,j}$\;
}
Return the plan of VMs.
}
\caption{\label{algo:IVP_algorithm}Incorporating VM Plan Online Algorithm}
\end{algorithm}

We use an example to explain this plan algorithm, as illustrated in Fig.~\ref{fig:behaviour_IVP}. Assume that at time $t=2$, there is one active server $server_1$ with $vm_1$ running on it, and there are one VM $vm_2$ arriving and two reserved VMs $\{vm_3,_vm_4\}$. As $t=2$ is the scheduling time, the plan algorithm is triggered. First each VM $vm_j{\in}\{vm_2,vm_3,vm_4\}$ is allocated on a virtual server $server_i{\in}\{server_2,server_3,server_4\}$, respectively. Then the profit $Profit_{i,j}$ is calculated according to Equation.~\ref{eq:profitCalc}. As $Profit_{1,3}$ is the maximum profit, the algorithm allocates the $vm_3$ to $server_1$ in the first round. This procedure repeats in the second round, and $vm_4$ is allocated on $server_1$. At last, it generates an allocation showed in Fig.~\ref{fig:behaviour_IVP}(b). Assume at time $t=7$, there is one VM $vm_5$ arriving. So the plan algorithm is triggered, and generates an allocation showed in Fig.~\ref{fig:behaviour_IVP}(d). Note that the allocation of reserved VM $vm_4$ is changed from $server_1$ to $server_2$.
\begin{figure}
\begin{tikzpicture}
    \draw[->, >=latex,line width=1pt] (0.0,0)--(3.75+0.05,0) node[pos=1,below] {$t$};
    \draw[->, >=latex,line width=1pt] (0,0.0)--(0,2);
    \draw[dotted] (0.0,1.75)--(0.0,0.0) node[below, scale=0.8] {$0$};
    \draw[dotted] (0.25,1.75)--(0.25,0.0);
    \draw[dotted] (0.5,1.75)--(0.5,0.0) node[below, scale=0.8] {$4$};
    \draw[dotted] (0.75,1.75)--(0.75,0.0);
    \draw[dotted] (1,1.75)--(1,0.0) node[below, scale=0.8] {$8$};
    \draw[dotted] (1.25,1.75)--(1.25,0.0);
    \draw[dotted] (1.5,1.75)--(1.5,0.0) node[below, scale=0.8] {$12$};

    \draw[dotted] (1.75,1.75)--(1.75,0.0) node[below=6pt, scale=0.8] {(a)$t_0=2$};
    \draw[dotted] (2,1.75)--(2,0.0) node[below, scale=0.8] {$16$};
    \draw[dotted] (2.25,1.75)--(2.25,0.0);
    \draw[dotted] (2.5,1.75)--(2.5,0.0) node[below, scale=0.8] {${20}$};
    \draw[dotted] (2.75,1.75)--(2.75,0.0);
    \draw[dotted] (3,1.75)--(3,0.0) node[below, scale=0.8] {${24}$};
    \draw[dotted] (3.25,1.75)--(3.25,0.0);
    \draw[dotted] (3.5,1.75)--(3.5,0) node[below, scale=0.8] {${28}$}; 
    
    \draw(0.0,1.5) rectangle (2.0,1.3) node[pos=0.55,text width=0.25em, scale=0.8] {$vm_{1}$};
        
    \draw(0.375,0.75) rectangle (2.75,0.55) node[pos=0.55,text width=0.25em, scale=0.8]{$vm_{3}$};
    
    \draw(0.25,0.55) rectangle (0.875,0.15) node[pos=0.55, scale=0.8]{$vm_{2}$};
    
    \draw(1.0,0.55) rectangle (2.5,0.30) node[pos=0.45,scale=0.8]{$vm_{4}$};
    
    \draw[dotted,line width=0.9pt](0.0,1.5) rectangle (2.75,0.9) node[right=0pt,pos=0.98,scale=0.8] {$server_1$};   
             
    \draw[->, >=latex,line width=1pt] (4.0,0)--(7.75+0.25,0) node[pos=1,below] {$t$};
    \draw[->, >=latex,line width=1pt] (4,0.0)--(4,2);
    \draw[dotted] (4.0,1.75)--(4.0,0.0) node[below, scale=0.8] {$0$};
    \draw[dotted] (4.25,1.75)--(4.25,0.0);
    \draw[dotted] (4.5,1.75)--(4.5,0.0) node[below, scale=0.8] {$4$};
    \draw[dotted] (4.75,1.75)--(4.75,0.0);
    \draw[dotted] (5,1.75)--(5,0.0) node[below, scale=0.8] {$8$};
    \draw[dotted] (5.25,1.75)--(5.25,0.0);
    \draw[dotted] (5.5,1.75)--(5.5,0.0) node[below, scale=0.8] {$12$};

    \draw[dotted] (5.75,1.75)--(5.75,0.0) node[below=6pt, scale=0.8] {(b)};
    \draw[dotted] (6,1.75)--(6,0.0) node[below, scale=0.8] {$16$};
    \draw[dotted] (6.25,1.75)--(6.25,0.0);
    \draw[dotted] (6.5,1.75)--(6.5,0.0) node[below, scale=0.8] {${20}$};
    \draw[dotted] (6.75,1.75)--(6.75,0.0);
    \draw[dotted] (7,1.75)--(7,0.0) node[below, scale=0.8] {${24}$};
    \draw[dotted] (7.25,1.75)--(7.25,0.0);
    \draw[dotted] (7.5,1.75)--(7.5,0.0) node[below, scale=0.8] {${28}$};  
    
    \draw(4.0,1.5) rectangle (6,1.3) node[pos=0.55,text width=0.25em, scale=0.8] {$vm_{1}$};
    
    \draw(4.375,1.3) rectangle (6.75,1.1) node[pos=0.55,text width=0.25em, scale=0.8]{$vm_{3}$};
    
    \draw(5.0,1.1) rectangle (6.5,0.9) node[pos=0.45,scale=0.8]{$vm_{4}$};
    
    \draw(4.25,0.75) rectangle (4.875,0.35) node[pos=0.55, scale=0.8]{$vm_{2}$};
    
    \draw[dotted,line width=0.9pt](4.0,1.5) rectangle (6.75,0.9) node[right=0pt,pos=0.98,scale=0.8] {$server_1$};    
    \draw[dotted,line width=0.9pt](4.0,0.75) rectangle (6.75,0.15) node[right=0pt,pos=0.98,scale=0.8] {$server_2$};    
  
    \draw[->, >=latex,line width=1pt] (0.0,-2.75)--(3.75+0.05,-2.75) node[pos=1,below] {$t$};
    \draw[->, >=latex,line width=1pt] (0,-2.75)--(0,-0.5);
    
    \draw[dotted] (0.0,1.75-2.5)--(0.0,0.0-2.75) node[below, scale=0.8] {$0$};
    \draw[dotted] (0.25,1.75-2.5)--(0.25,0.0-2.75);
    \draw[dotted] (0.5,1.75-2.5)--(0.5,0.0-2.75) node[below, scale=0.8] {$4$};
    \draw[dotted] (0.75,1.75-2.5)--(0.75,0.0-2.75);
    \draw[dotted] (1,1.75-2.5)--(1,0.0-2.75) node[below, scale=0.8] {$8$};
    \draw[dotted] (1.25,1.75-2.5)--(1.25,0.0-2.75);
    \draw[dotted] (1.5,1.75-2.5)--(1.5,0.0-2.75) node[below, scale=0.8] {$12$};

    \draw[dotted] (1.75,1.75-2.5)--(1.75,0-2.75) node[below=6pt, scale=0.8] {(c)$t_1=7$};
    \draw[dotted] (2,1.75-2.5)--(2,0.0-2.75) node[below, scale=0.8] {$16$};
    \draw[dotted] (2.25,1.75-2.5)--(2.25,0.0-2.75);
    \draw[dotted] (2.5,1.75-2.5)--(2.5,0.0-2.75) node[below, scale=0.8] {${20}$};
    \draw[dotted] (2.75,1.75-2.5)--(2.75,0.0-2.75);
    \draw[dotted] (3,1.75-2.5)--(3,0.0-2.75) node[below, scale=0.8] {${24}$};
    \draw[dotted] (3.25,1.75-2.5)--(3.25,0.0-2.75);
    \draw[dotted] (3.5,1.75-2.5)--(3.5,0-2.75) node[below, scale=0.8] {${28}$};

    \draw(0.0,1.5-2.75) rectangle (2,1.3-2.75) node[pos=0.55,text width=0.25em, scale=0.8] {$vm_{1}$};
    
    \draw(0.375,1.3-2.75) rectangle (2.75,1.1-2.75) node[pos=0.55,text width=0.25em, scale=0.8]{$vm_{3}$};
    
    \draw(1.0,1.1-2.75) rectangle (2.5,0.9-2.75) node[pos=0.45,scale=0.8]{$vm_{4}$};
    
    \draw(0.25,0.75-2.75) rectangle (0.875,0.35-2.75) node[pos=0.55, scale=0.8]{$vm_{2}$};
    
    \draw[dotted,line width=0.9pt](0.0,1.5-2.75) rectangle (2.75,0.9-2.75) node[right=0pt,pos=0.98,scale=0.8] {$server_1$};    
    \draw[dotted,line width=0.9pt](0.0,0.75-2.75) rectangle (2.75,0.15-2.75) node[right=0pt,pos=0.98,scale=0.8] {$server_2$};
    
    \draw(0.875,1.5+0.15+0.2-2.75) rectangle (2.75,1.5+0.15-2.75) node[pos=0.50,scale=0.8]{$vm_{5}$};

    \draw[->, >=latex,line width=1pt] (4.0,0-2.75)--(7.75+0.25,0-2.75) node[pos=1,below] {$t$};
    \draw[->, >=latex,line width=1pt] (4,0.0-2.75)--(4,2-2.5);
    \draw[dotted] (4.0,1.75-2.5)--(4.0,0.0-2.75) node[below, scale=0.8] {$0$};
    \draw[dotted] (4.25,1.75-2.5)--(4.25,0.0-2.75);
    \draw[dotted] (4.5,1.75-2.5)--(4.5,0.0-2.75) node[below, scale=0.8] {$4$};
    \draw[dotted] (4.75,1.75-2.5)--(4.75,0.0-2.75);
    \draw[dotted] (5,1.75-2.5)--(5,0.0-2.75) node[below, scale=0.8] {$8$};
    \draw[dotted] (5.25,1.75-2.5)--(5.25,0.0-2.75);
    \draw[dotted] (5.5,1.75-2.5)--(5.5,0.0-2.75) node[below, scale=0.8] {$12$};

    \draw[dotted] (5.75,1.75-2.5)--(5.75,0.0-2.75) node[below=6pt, scale=0.8] {(d)};
    \draw[dotted] (6,1.75-2.5)--(6,0.0-2.75) node[below, scale=0.8] {$16$};
    \draw[dotted] (6.25,1.75-2.5)--(6.25,0.0-2.75);
    \draw[dotted] (6.5,1.75-2.5)--(6.5,0.0-2.75) node[below, scale=0.8] {${20}$};
    \draw[dotted] (6.75,1.75-2.5)--(6.75,0.0-2.75);
    \draw[dotted] (7,1.75-2.5)--(7,0.0-2.75) node[below, scale=0.8] {${24}$};
    \draw[dotted] (7.25,1.75-2.5)--(7.25,0.0-2.75);
    \draw[dotted] (7.5,1.75-2.5)--(7.5,0.0-2.75) node[below, scale=0.8] {${28}$};  
    
    \draw(4.0,1.5-2.75) rectangle (6,1.3-2.75) node[pos=0.55,text width=0.25em, scale=0.8] {$vm_{1}$};
    
    \draw(4.375,1.3-2.75) rectangle (6.75,1.1-2.75) node[pos=0.55,text width=0.25em, scale=0.8]{$vm_{3}$};
    
    \draw(4.875,1.1-2.75) rectangle (6.75,0.9-2.75) node[pos=0.50,scale=0.8]{$vm_{5}$};
    
    \draw(5.0,0.75-2.75) rectangle (6.5,0.55-2.75) node[pos=0.45,scale=0.8]{$vm_{4}$};
    
    \draw(4.25,0.75-2.75) rectangle (4.875,0.35-2.75) node[pos=0.55, scale=0.8]{$vm_{2}$};
    
    \draw[dotted,line width=0.9pt](4.0,1.5-2.75) rectangle (6.75,0.9-2.75) node[right=0pt,pos=0.98,scale=0.8] {$server_1$};    
    \draw[dotted,line width=0.9pt](4.0,0.75-2.75) rectangle (6.75,0.15-2.75) node[right=0pt,pos=0.98,scale=0.8] {$server_2$}; 
    
\end{tikzpicture}
\caption{\label{fig:behaviour_IVP}Incorporating VMs batch arrival and VMs reservation scheduling. Assume each physical server has resource capacity of \{$\bm{1}$\} unit. $vm_1,vm_2,vm_3,vm_4,vm_5$ have resource capacity \{$\bm{\frac{1}{3}}$\}, \{$\bm{\frac{2}{3}}$\}, \{$\bm{\frac{1}{3}}$\}, \{$\bm{\frac{1}{3}}$\}, and \{$\bm{\frac{1}{3}}$\} unit respectively. The processing times for them are given by $p_1=16$, $p_2=5$, $p_3=19$, $p_4=12$, and $p_5=15$. The performance degradation factor $d_{21}$ is $0.25$, and other factors among them are $0$. (a): $t=2$ Initial servers and VMs. (b): $t=2$ Allocation result. (c): $t=7$ Initial servers and VMs. (d): $t=7$ Allocation result.}
\end{figure}
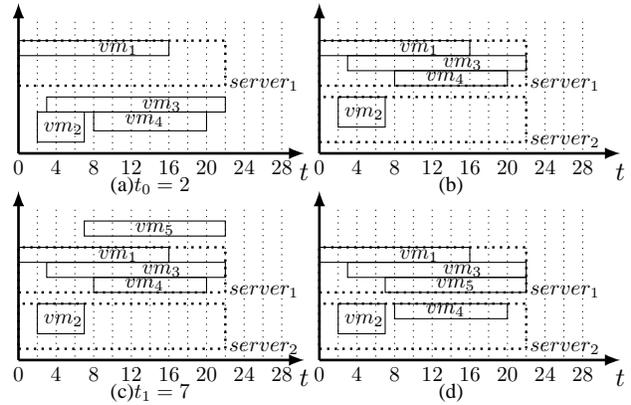

The time complexity of this algorithm is $O(mn^2)$ where $m$ is the number of servers ($\vert\mathcal{S}_{V}{\cup}\mathcal{S}_{I}\vert$) and $n$ is the number of VMs ($\vert\mathcal{N}_t{\cup}\mathcal{N}_t^f\vert$). This is followed from that each server $server_i$ proposes a profit to VM $vm_j$, and in each round we fix a VM $vm_j$.

\subsection{Distributed Design towards Data Center Scale}
In a large data center, it is time-consuming to gather the detailed information about each server and run the algorithm on a single server. We now propose a distribution scheme, which opposes to the algorithm introduced in the previous section that centralizes the information and picks up the best candidate server. Upon each new VM arriving at the data center, the information of VMs, which are waiting to be allocated, are passed to each active server. Each server $server_i$ maintains the information of VMs. Next, the algorithm proceeds in stages, and synchronizes using a common clock. In the first stage, a single client server $server_i$ proposes profits to the VMs and sends the maximum profit $maxProfit_{i,j}$ of $vm_j$ to the distribution server (dispatcher). The distribution server collects the maximum profits from all client servers and chooses the $max$-$maxProfit$, i.e., the current maximum benefit from the allocation of $vm_{\widehat{j}}$ on $server_{\widehat{i}}$. Then the allocation decision of stage $o$ is broadcast to client servers. In the subsequent stages, the client server $server_i$ receives the decision message $\langle{\widehat{i},\widehat{j}}\rangle$ from the distribution server, and proceeds the following two chooses: 1) If the profit of its $(o-1)$ stage is chosen, it fixes the $vm_{\widehat{j}}$ on it and removes the $vm_{\widehat{j}}$ from the unscheduled VMs. 2) If the profit of its $(o-1)$ stage is not chosen, it justly removes the $vm_{\widehat{j}}$ from the unscheduled VMs. The procedure of profit proposing is the same as the previous stage. Assume there are some VMs unscheduled, i.e., they are not benefit from being allocated to the client servers or cannot be allocated to current active servers. We run \textit{Algorithm.~\ref{algo:IVP_algorithm}} to schedule them with the input of these VMs and some current inactive servers. The pseudo-code of this algorithm is summarized in \textbf{\textit{Algorithm.~\ref{algo:TDPP_algorithm}}}.
\begin{algorithm}
\SetKwInOut{Input}{input}\SetKwInOut{Output}{output}
\Input{the set of VMs $\{\mathcal{N}_t{\cup}\mathcal{N}_t^f\}$ at time $t$}
\Output{the allocation results of VMs}
$Distribution Server:$\\
\Begin{
Initialized round: Broadcast $\mathcal{J}=\langle\{\mathcal{N}_t{\cup}\mathcal{N}_t^f\}\rangle$\;
\ForEach{round $o=1,2,...$}{
Receive message $\langle{maxProfit_{i,j}}\rangle$\;
\If{$\max\limits_{i,j}maxProfit_{i,j}{\geq}0$}
{
Pick $\langle{i,j}\rangle$=$\argmax\limits_{i,j}maxProfit_{i,j}$\;
Broadcast $\langle{i,j}\rangle$\;
Update $\mathcal{J}{\setminus}j$\;
}
}
\If{$\mathcal{J}{\neq}\emptyset$}
{
Run \textbf{\textit{Algorithm.~\ref{algo:IVP_algorithm}}}, i.e., $IVP(\mathcal{J},\emptyset)$;
}
Return the allocation of VMs.
}
$Client Servers:$\\
\Begin{
\ForEach{server $server_i$ in parallel}{
Receive message $\mathcal{O}$ from $Distribution Server$\;
\uIf{$\mathcal{O}=\langle\{\mathcal{N}_t{\cup}\mathcal{N}_t^f\}\rangle$}{
Save $\mathcal{J}=\mathcal{O}$\;
Set $\mathcal{S}_{V}$=$\{\mathcal{S}_j$=$\{vm_j\}|vm_j{\in}\{\mathcal{N}_t{\cup}\mathcal{N}_t^f\}\}$\;
Set $\mathcal{S}=\mathcal{S}_{V}{\cup}\{server_i\}$\;
Set $maxProfit_{i,j}=\max\limits_{i,j}Profit_{i,j}$\;
Send $\langle{maxProfit_{i,j}}\rangle$ to the $Distribution Server$\;
}
\uElseIf{$\mathcal{O}={\langle{\widehat{i},\widehat{j}}\rangle}$}{
\If{$\widehat{i}{==}server_i$}{
Fix $\widehat{j}$ on $server_i$ and Update $server_i$\;
}
Set $\mathcal{S}_{V}$=$\{\mathcal{S}_j$=$\{vm_j\}|vm_j{\in}\{\mathcal{J}{\setminus}{\widehat{j}}\}\}$\;
Set $\mathcal{S}=\mathcal{S}_{V}{\cup}\{server_i\}$\;
Set $maxProfit_{i,j}=\max\limits_{i,j}Profit_{i,j}$\;
Send $\langle{maxProfit_{i,j}}\rangle$ to the $Distribution Server$\;
}
}
}
\caption{\label{algo:TDPP_algorithm}VM Profit Plan Algorithm}
\end{algorithm}

\section{Performance Evaluation}
\label{sec:exp}
In this section, we study the performance of proposed algorithms on several comprehensive VM scheduling problems.

\subsection{Evaluation Setup}
\textit{Simulation Settings}: The simulations are run in a data center that equips physical servers with computing resource of $12$ Cores (E.g., \textit{HP ProLiant DL385 G6}). For simplicity, we assume that the servers are homogeneous in the data center and the VMs take up the total resource of their request demand. The configuration of VMs refers to the types of instances available in Amazon EC2 \cite{Amazon}. For example, a type VM, so called \textit{m1.small}, with $1$ Core computing units, $1.7$ GB memory and 160 GB storage space. As we focus on studying the arbitrating between energy consumption and performance interference degradation, we omit other resource bounds, such as memory, in the simulations. Four different types of VMs are available to be chosen with computing resource of $1$ Core, $2$ Cores, $4$ Cores, $8$ Cores, respectively. The simulations are conducted on two different scenarios corresponding to offline and online VMs scheduling. In the offline problem, there is a list of VMs that are waiting to be processed. In the online scenario, the VMs arrive randomly over time.

\textit{Compared Baseline Algorithms}: To provide benchmarks for our evaluations, we introduce other three algorithms:
\begin{itemize}
\item Random Strategy: it is a naive algorithm which randomly schedules the next VM on a physical server as long as the server has enough resource to host the VM.
\item Round Robin: it allocates the next VM on physical servers in turn, which is a used scheduling algorithm in Amazon EC2 \cite{Amazon}.
\item Minimum Increase Energy: it assigns the next VM to the server which minimizes the increment of energy consumption. Note that this algorithm is different from MIC, as MIC considers the increment of total cost. 
\end{itemize}

In all of the algorithms, it assumes that a new server will open if the next VM cannot be allocated on current active servers.

\textit{Parameters}: The power of a server is characterized by the three parameters of Equation.~\ref{eq:powerCalc}. We use HP ProLiant DL385 G6 with a 2 Chips/12 Cores processor. According to the server power consumption parameters, we set $P_{idle}=120W$ and $P_{busy}=258W$. The performance degradation cost is characterized by parameters: $\alpha$ and $\beta$. The parameter $\alpha$ represents the intensity of performance degradation penalty, and without loss of generality it is set to $15$. The tuning parameter $\beta$ is used for adjusting the energy consumption cost and performance degradation penalty, and also used for representing the weight between two costs.

\textit{Performance Metrics}: To evaluate the performance of proposed algorithms, we use the following four metrics:
\begin{itemize}
\item Normalized Total Energy Consumption: It shows the quality of the solution produced by the proposed algorithms in terms of total energy consumption.
\item Normalized Total Performance Degradation Penalty: This metric represents the penalties of the solution produced by the proposed algorithms in terms of performance degradation cost.
\item Normalized Total Cost: It is defined as the sum of energy consumption cost and performance degradation penalty.
\item Normalized Worst Degradation Factor: This indicates the worst performance degradation factor of VMs caused by the scheduling algorithms.
\end{itemize}

In the offline problem, we also record the total number of used physical servers and the makespan of the VMs.

\subsection{Evaluation of real workload}
We first conduct a small-scale experiment to evaluate the performance of the proposed algorithms. The performance degradation ratio is obtained from the statistics of SPECcpu 2006 benchmarks \cite{Kim13}. The properties (such as the arrival time, et.al) of the applications are drawn from a real OpenCloud Hadoop cluster trace \cite{Ren13}. The result is shown on Fig.~\ref{fig:eva_real_workload}. From the figure it can be seen that the overall cost are reduced apparently. More precisely, the MDC algorithm saves the total cost about 41\% compared with the BPV algorithm. This demonstrate the competitive advantages of the proposed algorithms against the methods which do not provide a unified consideration of both the energy consumption and the performance interference.
\begin{figure}[htbp]
\centering
\includegraphics[width=3.6in]{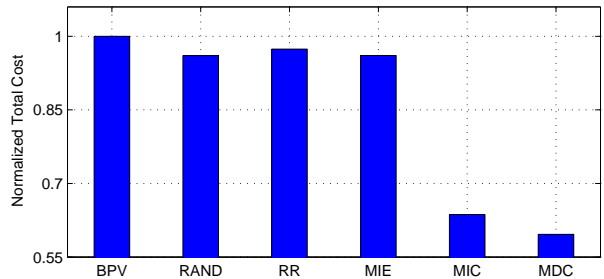}
\centering
\caption{\label{fig:eva_real_workload} Performance of algorithms on a real workload. (The result is normalized against BPV.)}
\end{figure}

\subsection{Evaluation of offline Problem Solution}
We now go to present our large-scale simulation results on algorithms proposed for offline problem. In this scenario, the computing resource required for a VM is uniformly chosen from four given types at random. We set each time slot as $1$ minute. The duration time of VMs is randomly generated from $[30,1000]$ and the arrival time of VMs is randomly generated from $[0,1000]$. The number of VMs varies from $100$ to $1000$ to emulate the low workload and heavy workload in the data center. It considers two kinds of degradation factors between VMs. One is generated from the normal distribution and the other is generated from the exponential distribution. For each simulation, we run proposed algorithms and compared algorithms using the same list of VMs, and the results of randomized algorithm is the average of running 10 times.

\textbf{Resource Violation:} In the first simulation of this evaluation, we examine the resource violation of algorithms when they do not take into account the performance degradation and do not update the duration time of VMs. In Fig.~\ref{fig:eva_resource_violation}, we show the resource violation of a server scheduled from the BPV algorithm. As we can see the computing resource surpasses the capacity during some periods (E.g., $t=653$ to $t=1469$) due to the stretch of VM execution time.
\begin{figure}[htbp]
\centering
\includegraphics[width=3.6in]{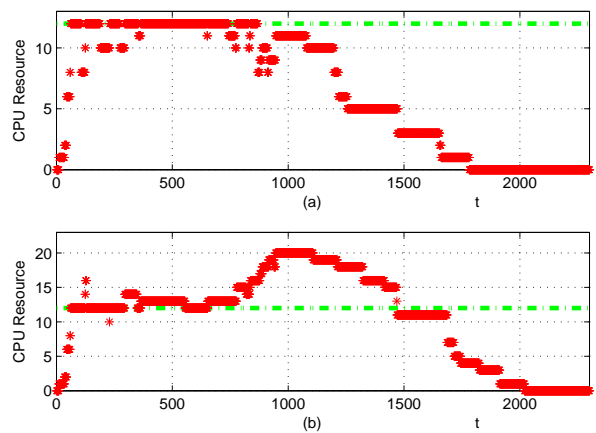}
\centering
\caption{\label{fig:eva_resource_violation}Resource violation of a server in BPV scheduling. The green dashed line means the CPU capacity of a server, i.e., 12 Cores. (a) CPU Resource used of a server from omitting the performance degradation; (b) The true CPU resource used when it considers the performance degradation stretching the execution time of VMs.}
\end{figure}

\textbf{Algorithm Performance Comparison:} We now discuss the performance of all above algorithms against the BPV algorithm with respect of four metrics, i.e., normalized energy consumption, performance degradation penalty, total cost and worst performance degradation factor. In this simulation, the degradation factor between VMs is generated from the normal distribution $N(0.0,0.2)$ $(d_{..}=0$ if $d_{..}<0)$. The results are depicted in Fig.~\ref{fig:eva_performance_algo}. As we can see, the performance of MDC algorithm is apparently better than other algorithms on reducing the performance degradation penalty. More precisely, the performance degradation penalty of MDC is $6\%$ against BPV, while MIC is $24\%$. Another observation from Fig.~\ref{fig:eva_performance_algo} is that MIC and MDC performs much better on worst performance degradation factor, which is important in data center due to SLA requirement. When considering the energy consumption, we find that MIC and MDC also have a slight reduction compared with other algorithms. This is because MIC and MDC reduce the unnecessary execution time due to the performance degradation causing by interference. As a result, MIC and MDC reduce the total cost of energy consumption and performance degradation penalty up to $62\%$ and $52\%$, respectively. It should be noted, however, the total cost of MDC is reduced $16\%$ when comparing with MIC.
\begin{figure*}[htbp]
\centering
\includegraphics[width=7.5in]{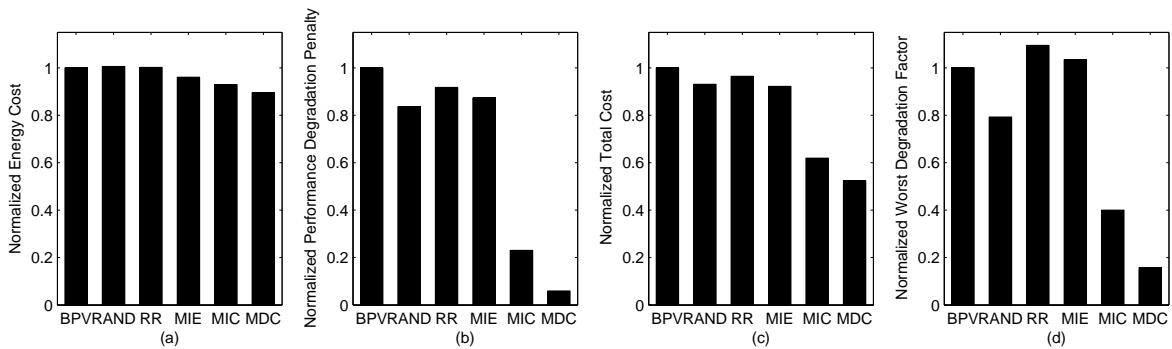}
\centering
\caption{\label{fig:eva_performance_algo}Performance of algorithms. (a) Normalized energy consumption against BPV; (b) Normalized performance degradation penalty against BPV; (c) Normalized total cost against BPV; (d) Normalized worst performance degradation factor against BPV.}
\end{figure*}

\textbf{Impact of Performance Degradation Factor:} We next investigate the impact of performance degradation factor on proposed algorithms. We keep the number of VMs fixed and execute the above algorithms with five different kinds of degradation factor between VMs. They are generated from the normal distributions $N(0.0,0.2)$, $N(0.0,0.4)$, $N(0.0,0.6)$, $N(0.0,0.8)$, $N(0.0,1.0)$, respectively $(d_{..}=0$ if $d_{..}<0)$. This corresponds to the interference between VMs is more fluctuation when the variance is changed from $0.2$ to $1.0$. In each normal distribution, we generate 5 groups of degradation factor, and we generate 3 groups of VMs list. The result is the average of the cross simulations, i.e., 15 times. Table~\ref{tb:eva_perf_degFactorN} shows how degradation factor affects the total cost. When the performance interference between VMs becomes more intensive, the total cost of BPV, RAND, RR and MIE algorithms become much larger, because they do not take into account the performance degradation penalty when making the scheduling decisions. For MIC and MDC algorithms, we can see that they lead to a slight total cost increment due to their intelligent scheduling. This rule also holds when we generate degradation factor between VMs from exponential distributions $E(100)$, $E(50)$, $E(20)$, $E(10)$, $E(5)$, $E(2)$, respectively. (The results are listed in Table~\ref{tb:eva_perf_degFactorE}.)
\begin{table}[htbp]
\centering
\caption{\label{tb:eva_perf_degFactorN}Algorithm performance in different $N(0.0, V)$ degradation factor, and the result is normalized against $V=0.2$.}
\begin{tabular}{c|c|c|c|c|c}\hline\hline
Algorithm&0.2&0.4&0.6&0.8&1.0\\\hline
BPV&1&186.83&${>}$200&${>}$200&${>}$200\\
RAND&1&1583.17&${>}$2000&${>}$2000&${>}$2000\\
RR&1&10.60&23.87&${>}$200&${>}$200\\
MIE&1&2.78&8.87&14.48&188.06\\\hline
MIC&1&1.0307&1.0333&1.0381&1.0406\\
MDC&1&1.0304&1.0480&1.0551&1.0567\\\hline\hline
\end{tabular}
\end{table}
\begin{table}[htbp]
\centering
\caption{\label{tb:eva_perf_degFactorE}Algorithm performance in different $E(\lambda)$ degradation factor, and the result is normalized against $\lambda=100$.}
\begin{tabular}{c|c|c|c|c|c|c}\hline\hline
Algo.&100&50&20&10&5&2\\\hline
BPV&1&1.0549&1.2728&2.1169&${>}$200&${>}$200\\
RAND&1&1.0447&1.2130&1.7677&230&${>}$300\\
RR&1&1.0415&1.1938&1.6070&13.5063&${>}$200\\
MIE&1&1.0437&1.2089&1.6300&3.6698&263.4512\\\hline
MIC&1&1.0471&1.1683&1.3006&1.4120&1.5360\\
MDC&1&1.0383&1.1179&1.2003&1.3072&1.4544\\\hline\hline
\end{tabular}
\end{table}

\textbf{Impact of Workload Density:} We compare the proposed algorithms on five data sets: $100$, $200$, $500$, $800$ and $1000$ VMs, and in each of them the arrival time of VMs is generated from the same range $[0,1000]$. I.e., the number of data set from $100$ to $1000$ represents the increment of workload density. The degradation factor between VMs is also generated from the normal distribution $N(0.0,0.2)$ $(d_{..}=0$ if $d_{..}<0)$. Fig.~\ref{fig:eva_perf_density} presents the results. In all simulations with different intensity, the minimum total cost is achieved by MDC due to its more global view. Moreover, the improvement margin is stable with the increment of workload density. MIC and MDC perform better in more intensive load. This is attributed to the fact that the other four algorithms would lead to more performance degradation when they do not take into account the performance interference in heavy load.
\begin{figure}[htbp]
\centering
\includegraphics[width=3.6in]{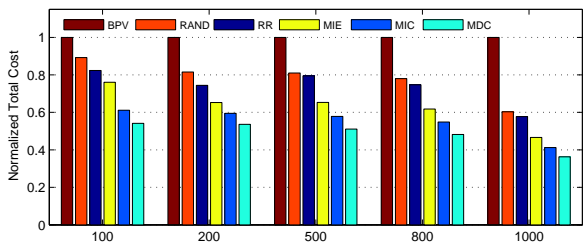}
\centering
\caption{\label{fig:eva_perf_density}Impact of different workload density.}\label{fig:graph}
\end{figure}
\begin{figure*}[htbp]
\centering
\includegraphics[width=7.5in]{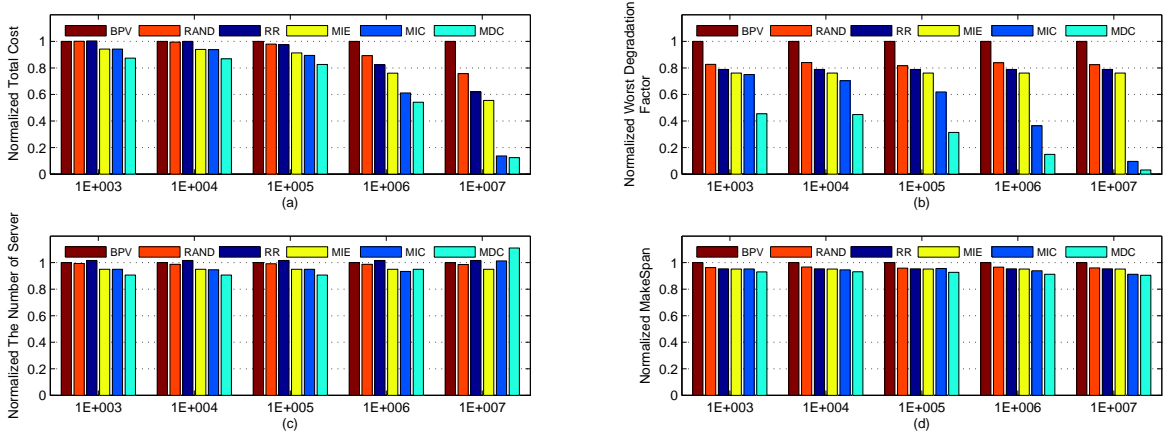}
\centering
\caption{\label{fig:eva_perf_weight}Impact of different weight $\beta$ between the operational cost and performance degradation penalty.}\label{fig:graph}
\end{figure*}

\textbf{Impact of Weight of Performance Degradation Penalty:} As mentioned in the optimization model of Section~\ref{sec:model}-C, the weight $\beta$ is some constant incorporating the normalization and the relative importance of the performance degradation penalty. Our virtual machine scheduling exploits this weight, and now we study the impact of this weight on the performance of the proposed algorithms. In this simulation, we focus on four metrics: total cost, worst performance degradation factor, the number of servers to be used and makespan. The results are depicted in Fig.~\ref{fig:eva_perf_weight}. The first observation from the Fig.~\ref{fig:eva_perf_weight} is that as the weight increases, the results of worst performance degradation factor, the number of servers to be used and makespan mainly stay the same in BPV, RAND, RR and MIE. This is because the scheduling of these four algorithms is not influenced by the weight. The results of their total cost increase only because the weight $\beta$ grows. Another remark is that the worst performance degradation factor reduces quickly in MIC and MDC when the weight surpasses a certain value (In our simulation, e.g., $\beta=10^6$). However, the number of servers to be used increases a little when we give more weight to performance degradation penalty, which leads to makespan metric reduction. Obviously, the improvement of total cost in MIC and MDC getting more with an increment of the weight.

\subsection{Evaluation of Online Algorithm}
We finally present the simulation on our online algorithms in the dynamic environment. It should be noted that the aforementioned algorithms except MDC can be transformed to the online versions, so the corresponding results are also held. We focus on evaluating the performance of \textit{Algorithm.~\ref{algo:IVP_algorithm}} which incorporates the VM batch arrival and VM reservation. In this simulation, we fix the total number of VMs and their performance degradation factors. We adopt different sizes of VMs to be revealed at each scheduling time to represent the VMs with batch arrival and reservation. There are $n_r$ VMs revealed at each scheduling time if the size is $n_r$. For example, it corresponds to one by one scheduling when the size is 1 (like OMIC), and there are 10 VMs revealed at each scheduling time when the size is 10. We run the algorithm with five randomly generated instances (1\#---5\#) and one sequentially generated instance (6\#) on each size. The result is depicted in Fig.~\ref{fig:eva_perf_online}. As we can see the performance is better with the the number of size increases, i.e., more information about VMs are revealed due to batch arrival and reservation. In addition, the total cost have a much improvement from the number of size 1 to 2. In summary, the \textit{Algorithm.~\ref{algo:IVP_algorithm}} exploits the properties of VMs scheduling in cloud data centers and obtains a better improvement of the total cost than OMIC algorithm which do not consider these properties.
\begin{figure}[htbp]
\centering
\includegraphics[width=3.6in]{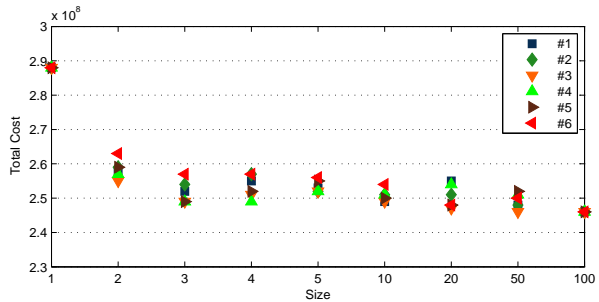}
\centering
\caption{\label{fig:eva_perf_online}Total cost of different size $n_r$ of VMs to be revealed at each scheduling time.}\label{fig:graph}
\end{figure}

\section{Conclusion}
\label{sec:conclusion}
In this paper, we present virtual machine scheduling for arbitrating between operational cost and performance interference in cloud data centers. While previous works only provide energy consumption management or performance interference optimization separately, we are among the first to build a joint model to capture the inherent tradeoff between the two contradictory objectives. We also develop efficient scheduling algorithms for both offline and online cases and improve them by exploiting some properties in clouds such as resource reservation. We evaluate the performance of the proposed algorithms by a comprehensive set of simulations. Our results confirm that a joint optimization that takes into account both VM combination and life-cycle overlapping can significantly reduce the operational cost, as well as the performance interference in cloud data centers.







%

\end{document}